\documentclass[preprint]{imsart}
\newcommand\independent{\protect\mathpalette{\protect\independenT}{\perp}}
\def\independenT#1#2{\mathrel{\rlap{$#1#2$}\mkern3mu{#1#2}}}

\usepackage[figuresright]{rotating}
\usepackage{amsmath}
\usepackage{amssymb}
\usepackage{bbm}
\usepackage{subfigure}
\usepackage{pstricks,pst-node,pst-plot}
\usepackage{color}
\usepackage{pst-all}
\usepackage{pspicture}
\usepackage{booktabs}
\usepackage{wasysym}
\usepackage[all]{xy}
\usepackage{epsfig,graphicx}
\usepackage{theorem}
\usepackage{natbib}
\usepackage{tikz}
\usetikzlibrary{positioning,decorations.markings,arrows,automata}
\usepackage{algorithm}
\usepackage{algorithmic}
\usepackage{multirow}
\usepackage{hyperref}

\newcounter{count}
\newcounter{asscount}

\newenvironment{lemma}[1][Lemma \arabic{count}]{\vspace{1em}\refstepcounter{count}\begin{trivlist}
\item[\hskip \labelsep {\bfseries #1}]\em}{\end{trivlist}\vspace{1em}}

\newenvironment{proof}[1][Proof]{\vspace{1em}\begin{trivlist}
\item[\hskip \labelsep {\bfseries #1}]}{\hfill$\Box$\end{trivlist}\vspace{1em}}

\newenvironment{definition}[1][Definition \arabic{count}]{\vspace{1em}\refstepcounter{count}\begin{trivlist}
\item[\hskip \labelsep {\bfseries #1}]}{\end{trivlist}\vspace{1em}}

\newenvironment{remark}[1][Remark]{\vspace{1em}\refstepcounter{count}\begin{trivlist}
\item[\hskip \labelsep {\bfseries #1}]}{\end{trivlist}\vspace{1em}}

\newenvironment{assumption}[1][Assumption \arabic{asscount}]{\vspace{1em}\refstepcounter{asscount}\begin{trivlist}
\item[\hskip \labelsep {\bfseries #1}]}{\end{trivlist}\vspace{1em}}

\begin{document}
%---------------------------------------------------
\begin{frontmatter}

\title{Towards personalised intervention:\\ A causal-dynamical framework to determine psychological treatment trajectories}

\runtitle{Simulating interventions}
\author{Lourens Waldorp$^1$ }\ead{waldorp@uva.nl}
\author{Titus M\"{u}rtz$^1$} \author{Anita Jansen$^2$}
\author{Jonas Haslbeck$^{1,2}$}\ead{haslbeck@protonmail.com}
\address{$^1$ University of Amsterdam, Nieuwe Achtergracht 129-B, 1001NK, the Netherlands}
\address{$^2$ Maastricht University, Universiteitssingel 40, 6229 ER, the Netherlands}
%\ead[label=e1]{waldorp@uva.nl}
\runauthor{Waldorp et al.}
%\texttt{\em waldorp@uva.nl}\hspace{2em}\texttt{\em jonashaslbeck@protonmail.com}\\

%\printead{e1}

\begin{abstract}
For approximately half of the individuals receiving mental health care, the results are suboptimal, even when treatments align with evidence-based guidelines. These limited effects may partly stem from how clinical decisions on treatment focus are made in mental health care. Typically, treatment strategy is guided by the diagnostic classification combined with the individualized case conceptualization. While standard, this approach may fall short for several reasons such as biases on the part of both the patient and therapist, and treatment guidelines being based on average effects that may not (exactly) suit the individual patient. To address these challenges, we propose a novel framework that reduces biases in clinical decision-making and makes it genuinely possible to tailor treatment focus to the individual patient. This framework involves (a) constructing causal graphs and estimating causal effects from intensively collected, longitudinal patient data, (b) simulating new time series based upon the causal relationships, and (c) using these simulations to identify the most effective treatment focus for the individual patient. By simulating and comparing different intervention strategies and examining both the estimated individual’s responsiveness and its long-term effectiveness, this approach may generate useful insights to guide treatment focus and strategy, which can lead to a significant improvement of treatment outcomes in mental health care.  
\end{abstract}


\end{frontmatter}
%-------------------------------------------------------

%---------------------------------------------------------------------------
\section{Introduction}
Despite the wide range of evidence-based interventions, outcomes of treatments in mental health care remain suboptimal for a substantial proportion of patients \citep{Leichsenring:2022}. Less than half of the individuals receiving treatment achieve sufficient improvement, even when interventions are delivered in accordance with the established evidence-based clinical guidelines \citep{Cuijpers:2024,Cuijpers:2026,Hoffmann:2012}. This could indicate a significant limitation of current approaches to selecting the appropriate treatment focus. 
In everyday practice, clinical decision-making regarding the appropriate treatment focus is typically guided by a combination of diagnostic classification and individualized case conceptualization \citep{Eells:2025}. Diagnoses, often based on systems such as the diagnostic statistical manual 5 (DSM-5), are linked to evidence-based treatment recommendations derived from group-level research. These recommendations are then refined through the individual case conceptualization, in which clinicians integrate information about the patient’s history, symptoms, and context to tailor the evidence-based treatment. Case formulations are a combination of scientific knowledge, clinical expertise, and lived experience – they are considered to be an essential therapeutic component to guide effective treatment \citep{Eells:2025}. 

Although this approach of clinical decision making is widely accepted and grounded in scientific evidence, it has notable limitations. First, treatment guidelines are based on average effects observed in heterogeneous populations and may not adequately capture individual variability in treatment response. Secondly, case formulations—while valuable—are complex, inherently subjective, and likely to be biased, which can lead to poor decision-making \citep{Lutz:2025}. A clinician might apply heuristics, which tap into information about which treatments seemed to work for most patients. Additionally, it is hard to identify feedback relations, which are considered the main drivers in the persistence of mental illness \citep{Hayes:2020,Evers:2026}. Both the diagnosis and the case formulation often present a relatively static picture of the psychopathology, meaning that they offer only limited insight into the day-to-day dynamics of the symptoms. 
To address these limitations of current clinical decision-making, there is growing interest in more personalized, data-driven approaches to diagnosis and the associated selection of treatment strategies. A promising approach involves the use of extensive longitudinal data, collected through repeated measurements in patients’ daily lives, to gain a better understanding of the dynamics of individual symptoms. Using such data, connections can be established between symptoms, feelings, behavior, beliefs, and contextual factors over time \citep{Roefs:2022}. In this way, we move beyond static and potentially biased descriptions, to arrive at a more data-driven and process-oriented understanding of mental disorders which may improve clinical decision-making on treatment focus. 
The use of more formal data-driven methods for inference about the complex multicomponent nature of the individual’s psychopathology, might improve the response rates of therapy \citep{Hofmann:2019,Hitchcock:2022}. There is some evidence to suggest that data-driven personalisation of treatment trajectories is beneficial \citep{Scholten:2022,Delgadillo:2022,Hofmann:2019}. 
However, even if such information were available, currently, there is no framework to properly test, prior to treatment, any hypotheses about which treatment focus might work best or is most efficient. 

Here, we propose a framework and a method for improving clinical decision making regarding the most appropriate treatment strategy for the individual patient.  
By modelling observed variables as a causal network and then performing interventions on a dynamical version of that network in simulations,
we can test which of the intervention strategies is predicted to be the most effective. Our framework is akin to those in weather prediction \citep{Leutbecher:2008,Palmer:2019}, where simulations are used to determine what could happen next in what-if situations, and our framework is also similar to those in physics \citep{Sethna2004}, biology \citep{Li:2022} and behavioural sciences
\citep{Nishi:2020}, where simulations are used to test the effects of interventions. 
These methods, simulations to determine which intervention is optimal, have led to great advances in weather prediction  \citep{Allen:2025} and in determining climate policy \citep{Hazeleger:2024}. Our proposal differs with respect to existing efforts in that our framework:
($i$) It uses causal information that leads to stable inference \citep{Pearl:2000,Peters:2017,Chauhan:2025}, and  the interpretation of interventions matches knowledge from experiments \citep{Woodward:2003,Gillies:2019,Mooij:2020a}, ($ii$) effectively uses feedback relations that appear crucial in the persistence of symptoms \citep{Wittenborn:2016} and learning about effective interventions \citep{Hayes:2020,Hofmann:2019}, and ($iii$) invokes non-linear modelling to emulate characteristic features of real life processes  \citep{Hayes:1998,Schiepek:2019,Waldorp:2020}. We explain each of these advantages of our frameworek. 

($i$) 
%Causal relations are at the heart of many sciences \citep{Gillies:2019,Glymour:2019} and are beginning to play a larger role in machine learning. 
While data-driven methods like computational psychiatry \citep{Teufel:2016} and AI/machine learning aim to improve treatments in mental health \citep{Lunansky:2022,Cui:2023,Fischer:2025,Ryan:2025c}, they seem unable to obtain an increase in treatment response rates \citep{Chauhan:2025}. Current AI and machine learning apply correlational methods \citep[][except, e.g., \citet{Ryan:2025b}]{Garg:2022,Ntekouli:2024,Zhang:2025}, which might be able to predict well, but fail to evaluate the effects of interventions \citep{Chauhan:2025}. For example, the conclusion that treatment leads to lower remission rates, is likely to be confounded by severity; severity leading both to individuals seeking treatment and to lower remission rates (i.e., $T\leftarrow S\to R$ and $T\to R$). Such scenarios reveal that it is crucial to determine the causal structure, which leads to stable (reliable) inference \citep{Spirtes:1993,Pearl:2000,Hitchcock:2022,Chauhan:2025}.  Causal information is important because: ($a$) how variables interact may provide an explanation about why problems persist. For example, someone with social anxiety might start avoiding social situations, which exacerbates problems at work and increases their anxiety, causing their social anxiety to worsen even further (i.e., $A\to E \to M\to A$). Learning about such a causal process may benefit treatment, which is also the basic principle of cognitive behavioural theorapy (CBT) models. The main difference is that the framework we propose draws on empirical data collected from individual patients, whereas (CBT-based) case conceptualisations are primarily derived from clinical interviews and clinical reasoning. And ($b$) the causal framework encompasses \textit{what if} questions, e.g., in the above scenario, if we reduce avoidance,  this will lead to more social interactions.   Such processes are crucial to both the case conceptualization and determination of the treatment focus. 

($ii$) Data-driven methods do not necessarily consider the feedback loops that are crucial \citep{Cui:2023} to understand the mechanisms that drive the persistence of mental illness, although such mechanisms are not necessarily the cause \citep{Hayes:2015b, Scheffer:2024b,Scheffer:2024}. In the example above, $A\to E \to M\to A$, intervening on avoidance may be achievable, while intervening on anxiety may be more problematic. This shows that knowing these feedback relations is critical for selecting optimal interventions. Although it is not easy to obtain feedback relations \citep{Park:2024, Waldorp:2024b}, this can be improved by obtaining repeated measures or time series \citep{Hyttinen:2016}. We therefore include ecological momentary assessment data \citep{Shiffman:2008} to obtain causal effect estimates. As an example, we found evidence of a feedback relation in real data of a patient: {\em anxiety $\leftrightarrow$ stressed $\to$ sad $\to$ guilty $\to$ anxiety} (see Figure \ref{fig:causal-effects-p1}).

($iii$) Non-linear modelling refers to concepts such as multiple stable states \citep[e.g., ill and healthy states,][]{Leemput:2014,Waldorp:2020} and tipping points \citep[i.e., points where a small perturbation will lead to a different stable state,][]{Scheffer:2001}, but also sudden transitions from one stable state to another, and hysteresis \citep[e.g., getting back to as it was before the (mental) illness is more difficult than getting (mentally) ill,][]{Cramer:2016,Maas:2020}. Because it appears that such aspects are important parts of psychopathology \citep{Hayes:1998,Hosenfeld:2015,Cramer:2016,Hayes:2020,Hofmann:2025},  it seems imperative to take such phenomena into account in modelling \citep{Hofmann:2019}. 
%Another type of feedback is between the clinician and patient, which is imperative to understand because treatment is an interactive process \citep{Shapiro:2015}.

We start in Section  \ref{sec:simulating-testing} by giving a conceptual overview of our proposed framework, in which we explain how to get from time series data to causal effects estimates, then use non-linear modelling to simulate new time series, and finally determine intervention foci and evaluate the intervention foci. Then in Sections \ref{sec:causal-graph} to  \ref{sec:interventions-comparisons} we give details and illustrations on each of these steps (some technical details have been deferred to the Appendix). Finally in Section \ref{sec:discussion-conclusion} we describe how the framework could be extended in multiple ways, to do justice to the complexity and nuances of real patients. 

\section{Simulating and testing psychological treatment trajectories}\label{sec:simulating-testing}

The framework involves several different steps to get from conceptual ideas of influence between relevant variables to more numerical and predictive versions for comparisons of intervention starting points. Our framework consists of the following steps, which are described below:
\begin{itemize}
    \item[($a$)] obtain a causal network and causal effects, 
    \item[($b$)] transform to nonlinear dynamic Markov process, and
    \item[($c$)] apply and compare starting points of interventions to decide on treatment strategy.
\end{itemize}
An overview of these three steps is provided in Figure \ref{fig:flow-diagram}. In what follows, we make particular choices for ($a$)-($c$) so we can illustrate our framework. 
%We do not think that those choices are optimal for all situations. Instead, our focus is to present a general framework and we hope that the specific choices can be improved and adapted to various context in future research.

%
\begin{figure}[H]\centering
    \pgfimage[width=\textwidth]{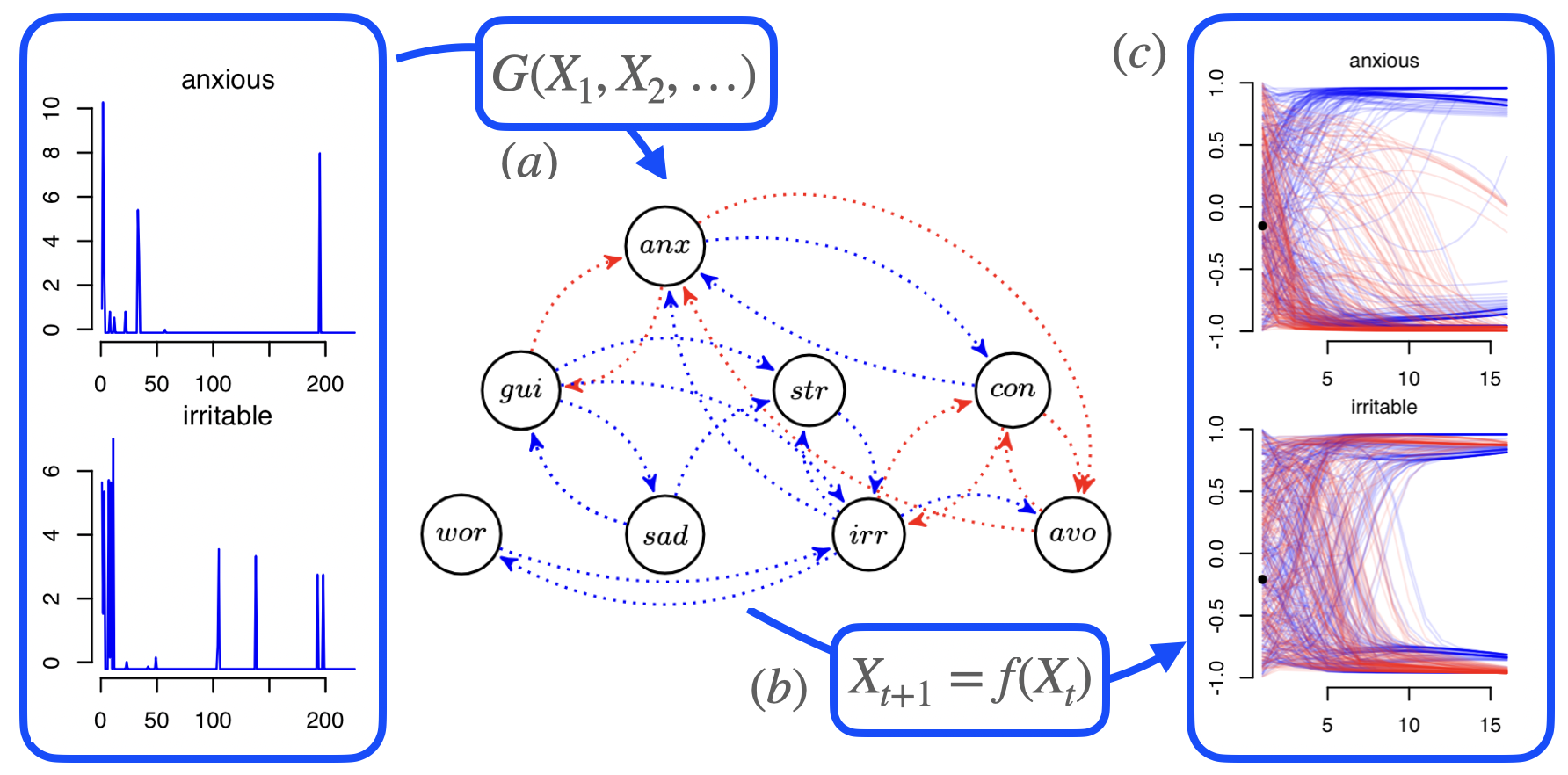}
\caption{Flow dagram of the proposed framework to get from data to possible starting points of interventions. The left panel shows a realisation of EMA data of a real patient with a mental disorder for two variables, \textit{anxious} and \textit{irritable}.  In (a) we take the EMA data and make a causal effects graph based on the causal graph. Then in (b) we use the EMA data to obtain the causal effects of this particular patient to make a non-linear Markov process to simulate possible trajectories for a control (blue) and intervention (red) condition. In (c) are the simulated temporal processes for the control and intervention conditions, which can be compared to determine the largest effect over time. Boxed material signifies it is possible to exchange this part for different data or a different procedure.  }
\label{fig:flow-diagram}
\end{figure}

\textit{Ad} ($a$). Our objective is to answer \textit{what if} questions about a possible starting point of treatment given specific information and context about an individual. To achieve this, we require a causal model, a causal network or graph of relevant variables, and temporal information will help in estimating the causal model. Typically, we will use ecological momentary assessment (EMA, see Figure \ref{fig:flow-diagram}), where several times a day a few questions are answered about emotions, behaviours, cognitions, and contexts (e.g., social environment) to capture the current state of an individual in their natural environment \citep{Shiffman:2008,Kossakowski:2019, Schiepek:2022}. This provides a wealth of data measured outside the treatment sessions of the patient and can be obtained before the patient starts with the actual treatment. Although perhaps not necessary, these data can be augmented (or replaced) by patient-clinician concepts of the issues at hand \citep[e.g.,][]{Hofmann:2025}, or perceived causal networks, where patient and/or clinician make a causal network of the relevant variables \citep{Frewen:2012,Burger:2024}.
The data can be used to obtain a causal model and measures of causal effects \citep[($a$) in Figure \ref{fig:flow-diagram},][]{Maathuis:2009}. A causal model contains nodes (variables) and their connections, and causal effects provide estimates of the impact of an intervention of one or several nodes on all other nodes \citep{Nandy:2017,Waldorp:2024b}. Obtaining measurements before and during treatment will improve our causal models and estimates of feedback relations \citep{Waldorp:2024b}. 

\textit{Ad} ($b$). We can use these causal models as models for a patient and make predictions about the response to different starting points of treatments. A dynamical description is obtained by using an Ising mean field function (hyperbolic tangent) with coupling between the variables derived from the causal model \citep[Figure \ref{fig:flow-diagram}($(b)$),][]{Maas:2020,Kohler:2021}. In the current version, there are two attractor states for each of the variables, such that each variable ends up at low or high values \citep[][]{Brunton:2016,Ntekouli:2024}. Such a model aligns well with some of the well-known phenomena that may or may not be present in psychopathology, such as bistability \citep{Hosenfeld:2015,Kossakowski:2019b}, hysteresis \citep{Cramer:2016}, and tipping points \citep{Leemput:2014,Scheffer:2024}. It is not our intention to suggest that non-linear phenomena are always present, just that a model should in principle be able to reflect such non-linearities, should they be present. In general, change, leading to bistability etc, will be useful to make a case formulation in terms of a landscape of the intricate interplay of emotions, cognitions, behaviours and contexts \citep{Schiepek:2019,Hayes:2020,Moulder:2021}, which may indicate the focus of the treatment \citep{Tschacher:2019,Hofmann:2025}.
%Of course, other models could be equally valid, like the Blume-Capel mean field model which has three stable states \citep{Maas:2026,Finnemann:2026}, and this should be tested. 
%Such temporal processes can be obtained in different ways, and, in fact, I will use this fact to my advantage by considering converging evidence across different versions of the processes \citep{Palmer:2023}. 

\textit{Ad} ($c$). In our framework we determine which intervention seems most effective and efficient \citep{Heinze:2018b}. This raises the question how interventions are acting on the dynamic model we use to approximate the system representing the disorder at hand. How to do this is not obvious.  Interventions can be conceptualized as perturbations on nodes, edges \citep{Valente:2012}, or other features of the model \citep[e.g., disturbance parameters,][]{Peters:2016}, in many different ways. For instance, we can imitate a stress reduction process by repeatedly acting on certain nodes (e.g., reducing anxiousness) or edges (e.g., reducing the connection between anxiousness and avoidance) in a causal model (see Figure \ref{fig:flow-diagram}(c)). An intervention can be performed in different ways, single or multiple variables simultaneously, including changes in connections (see Section \ref{sec:interventions-comparisons}).
To come to a possible treatment trajectory, we compare the different focus points of interventions in terms of how and when they reduce the levels on all variables simultaneously compared to a control where no intervention is applied (see Figure \ref{fig:flow-diagram}($c$)). This choice of measure of success is based on the idea that many different factors define whether a treatment is successful \citep{Schiepek:2017} and current practice \citep{Hayes:2020}.

\section{Causal graph and effects}\label{sec:causal-graph}
We introduce some concepts of a causal graph and describe how to relate the causal structure from a graph to a probability distribution, which can be estimated from data. Then we use the obtained causal graph to estimate causal effects, which reflect intervention effects. We remain informal here and defer more technical details to Appendices \ref{app:causal-basics} and \ref{app:assumptions-causal-discovery} and the cited references. 

The objective of causal inference is to obtain estimates of how the variables affect each other. To this end, let $\mathcal{G}$ be a graph with nodes (vertices) $V=\{1,2,\ldots,p\}$ and edges $E$ between the nodes.  We use two different types of edges: a directed edge $i\to j$ means that $i$ causes $j$, and a bidirected edge $i\leftrightarrow j$ means that there is uncertainty about the causal relation, which might be feedback or there might be a confound. 
We do not know the causal relations between nodes, but we can learn them from data through conditional (in)dependence relations for the associated random variables $X_i$ for all nodes. If we find that $X_i$ is conditionally independent of $X_j$ given $X_k$, then we conclude that there is no edge between $i$ and $j$. By considering many conditional independence relations, it is possible to determine how variables affect one another \citep{Spirtes:1993,Pearl:2000}. However, we will not be able to learn all connections or directions, and so we end up with a mixed graph of edges $i\to j$ or $i\leftrightarrow j$, where the latter suggests that from data alone, it cannot be determined which direction the edge should have or whether there might be a latent confound.  The algorithm we have used for this is fast causal inference \citep[FCI,][]{Spirtes:1996,Kalisch:2007}.

Figure \ref{fig:network-p1} shows the result of FCI applied to patient EMA data at zero lag (see Appendix \ref{app:causal-basics} and \ref{app:analyses-R} for details and an example assuming lag one relations over time). Figure \ref{fig:network-p1} shows that there might be evidence for feedback (cycles): $\textit{anxiety} \leftrightarrow \textit{stressed} \leftrightarrow \textit{worry}$. Such cycles could be relevant to determine intervention foci \citep{Hayes:2020}. 
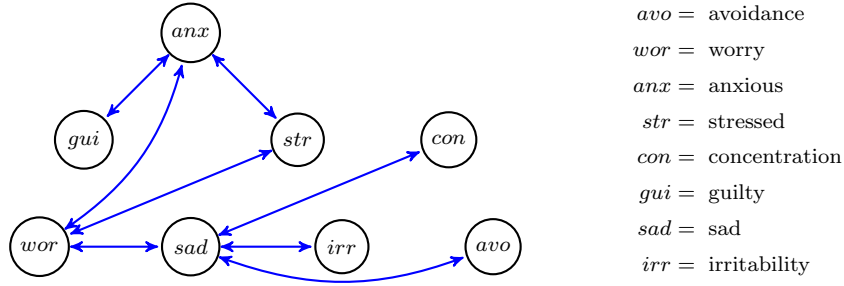
\begin{figure}[t]
\begin{tabular}{c @{\hspace{5em}} p{0.4\textwidth}}
\raisebox{-10em}{
\begin{tikzpicture}[-, >=stealth',shorten >=1pt,auto,node distance=2cm,
  thick,main node/.style={circle,fill=blue!10,draw,font=\footnotesize\sffamily}]
  \tikzstyle{sample} = [circle,fill=blue!20,draw,font=\footnotesize\sffamily]
  \tikzstyle{noSample} = [circle,fill=red!0,draw,font=\footnotesize\sffamily]
  \tikzstyle{sampleEdge} = [font=\sffamily\small,blue]
  \tikzstyle{noSampleEdge} = [node/.style={font=\sffamily\small},red]	
	
  \node[noSample] (anx) [] {$anx$};
  \node[noSample] (gui) [below left of=anx] {$gui$};
  \node[noSample] (str) [below right of=anx] {$str$};
  \node[noSample] (con) [right of=str] {$con$};
  \node[noSample] (sad) [below right of=gui] {$sad$};
  \node[noSample] (irr) [right of=sad] {$irr$};
  \node[noSample] (avo) [right of=irr] {$avo$};  
  \node[noSample] (wor) [left of=sad] {$wor$};

 \path[sampleEdge,<->]
    (anx) edge [right] node [xshift=0.2em] {} (str)  %above, yshift=0em, xshift=0.4em
    (str) edge [right] node [above] {} (wor)
    (wor) edge [right] node [] {} (sad)
    (wor) edge [bend right=20] node [below] {} (anx)
    (sad) edge [right] node [below] {} (irr)
    (sad) edge [bend right=20] node [] {} (avo)
    (sad) edge [right] node [left, yshift=0.5em, xshift=0.3em] {} (con)
    (anx) edge [right] node [above] {} (gui);

\end{tikzpicture}
}
&
\vspace{-3em}
\begin{itemize}
\item[$avo =$] avoidance 
\item[$wor =$] worry 
\item[$anx =$] anxious 
\item[$str =$] stressed 
\item[$con =$] concentration
\item[$gui =$] guilty
\item[$sad =$] sad
\item[$irr =$] irritability
\end{itemize}

\end{tabular}
\caption{Graph of the eight variables; result of applying FCI to the time series data using lag zero. A directed edge means there is evidence for a causal relation (there are none here), and a bidirected edge, where it is unclear what the relation is. }
\label{fig:network-p1}
\end{figure}

\subsection{Causal effects}

The causal graph provides a suggestion of possible effects, but it is hard to ``read-off'' the causal effects, the impact of an intervention on one of the nodes on all other nodes, for example. It is more informative to learn about the impact of a change in one variable on all other variables, i.e., the effect of a do-intervention of one node on any of the other nodes \citep{Pearl:2000}. In a do-intervention on node $i$, all of the incoming arrows to node $i$ are removed, so that full control is exerted over node $i$; the value of $X_i$ is determined only by the intervention \citep[][see Appendix \ref{app:assumptions-causal-discovery}]{Spirtes:2000}. For example, if in the graph $\textit{worry}\to\textit{anxious}\to\textit{guilty}$ we were to apply a do-intervention on \textit{anxious} we remove the effect $\textit{worry}\to \textit{anxious}$, but $\textit{anxious}\to \textit{guilty}$ remains. The distribution of the variables is then changed according to the do-intervention by excluding the incoming arrow on \textit{anxious} \citep[][and see Appendix \ref{app:assumptions-causal-discovery}]{Lauritzen:2002,Eberhardt:2007}. In the Gaussian case, when comparing two do-interventions on node $k$, we obtain the average causal effect 
\begin{align}\label{eq:average-causal-effect}
\mathbb{E}(X_i\mid \text{do}(X_k=x_k +1))-\mathbb{E}(X_i\mid \text{do}(X_k=x_k)),
\end{align}
which does not depend on the value $x_k$. Note that we marginalise over the other variables (covariate adjustment), so the average effect is determined over all the paths from $k$ to $i$ in (\ref{eq:average-causal-effect}). 
%See Appendix \ref{app:regression-causality} for an illustration of the difference between regression and causality.

In the absence of a definitive causal graph, such causal effects must be considered with respect to some uncertainty about the directions of the arrows. The estimated graph with FCI represents such uncertainty by referring to a set of graphs that comply with the output of FCI \citep{Richardson:2002,Mooij:2020a}, i.e. the FCI output is an equivalence class with many different versions of a causal graph. Hence,  \citet{Maathuis:2009} propose to estimate the causal effect on all possible versions in the Markov equivalent class. Subsequently, for each of the causal models, the causal effect is computed (regression coefficients but conditioned only on the parents, see Appendix \ref{app:regression-causality}), and the set of possible values provides a range of values that take into account the uncertainty we have with respect to all equivalent models in the same Markov class. We choose the minimal value for the causal effects here because it is a lower bound of the causal effect, but other possibilities (like the  average) might also be useful. Figure \ref{fig:causal-effects-p1} shows these causal effects, using dotted arrows, for the data of a patient for the causal effects that are at least of absolute size 0.05 (the causal effect graph of five patients are shown in Appendix \ref{app:results-5patients}). 

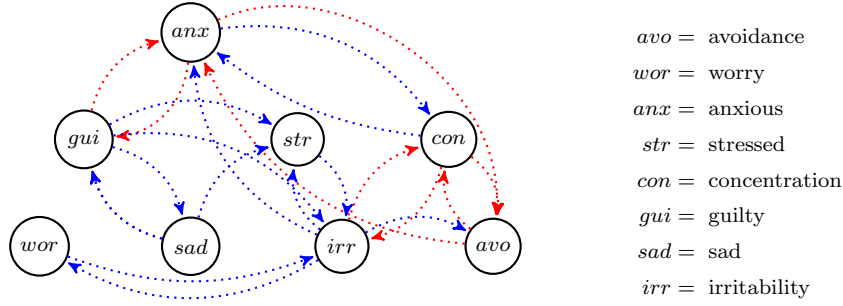
\begin{figure}[t]
\begin{tabular}{c @{\hspace{5em}} p{0.4\textwidth}}
\raisebox{-10em}{
\begin{tikzpicture}[-, >=stealth',shorten >=1pt,auto,node distance=2cm,
  thick,main node/.style={circle,fill=blue!10,draw,font=\footnotesize\sffamily}]
  \tikzstyle{sample} = [circle,fill=blue!20,draw,font=\footnotesize\sffamily]
  \tikzstyle{noSample} = [circle,fill=red!0,draw,font=\footnotesize\sffamily]
  \tikzstyle{sampleEdge} = [font=\sffamily\small,blue]
  \tikzstyle{noSampleEdge} = [node/.style={font=\sffamily\small},red]	
	
  \node[noSample] (anx) [] {$anx$};
  \node[noSample] (gui) [below left of=anx] {$gui$};
  \node[noSample] (str) [below right of=anx] {$str$};
  \node[noSample] (con) [right of=str] {$con$};
  \node[noSample] (sad) [below right of=gui] {$sad$};
  \node[noSample] (irr) [right of=sad] {$irr$};
  \node[noSample] (avo) [right of=irr] {$avo$};  
  \node[noSample] (wor) [left of=sad] {$wor$};

\path[sampleEdge,->,dotted,red]
 (avo) edge [bend left=30] node [xshift=0.2em] {} (anx)
 (avo) edge [bend left=30] node [xshift=0.2em] {} (con)
 (anx) edge [bend left=60] node [xshift=0.2em] {} (avo)
 (anx) edge [bend left=40] node [xshift=0.2em] {} (gui)
 (con) edge [bend left=30] node [xshift=0.2em] {} (avo)
 (con) edge [bend left=30] node [xshift=0.2em] {} (irr)
 (gui) edge [bend left=30] node [xshift=0.2em] {} (anx)
 (irr) edge [bend left=30] node [xshift=0.2em] {} (con);
 %above, yshift=0em, xshift=0.4em

 \path[sampleEdge,->,dotted]
 (sad) edge [bend left] node [xshift=0.2em] {} (gui)
 (anx) edge [bend left] node [xshift=0.2em] {} (con)
 (sad) edge [bend left] node [xshift=0.2em] {} (str)
 (sad) edge [bend left] node [xshift=0.2em] {} (gui)
 (wor) edge [bend right=20] node [xshift=0.2em] {} (irr)
 (str) edge [bend left=30] node [xshift=0.2em] {} (irr)
 (con) edge [bend left=15] node [xshift=0.2em] {} (anx)
 (gui) edge [bend left=30] node [xshift=0.2em] {} (sad)
 (gui) edge [bend left=30] node [xshift=0.2em] {} (str)
 (gui) edge [bend left=30] node [xshift=0.2em] {} (irr)
 (irr) edge [bend left=30] node [xshift=0.2em] {} (wor)
 (irr) edge [bend left=30] node [xshift=0.2em] {} (avo)
 (irr) edge [bend left=30] node [xshift=0.2em] {} (anx)
 (irr) edge [bend left=30] node [xshift=0.2em] {} (str)
 ;
 %above, yshift=0em, xshift=0.4em
    
\end{tikzpicture}
}
&
\vspace{-3em}
\begin{itemize}
\item[$avo =$] avoidance 
\item[$wor =$] worry 
\item[$anx =$] anxious 
\item[$str =$] stressed 
\item[$con =$] concentration
\item[$gui =$] guilty
\item[$sad =$] sad
\item[$irr =$] irritability
\end{itemize}

\end{tabular}
\caption{Graph of the eight variables and the (total) causal effects represented by the dotted arrows, estimated by the \texttt{\em ida} package \citep{Maathuis:2009}. Blue arrows refer to positive effects and red arrows refer to negative effects. Only the causal effects $z$ with $|z|\ge 0.05$ are shown. }
\label{fig:causal-effects-p1}
\end{figure}

Both positive (blue) and negative (red) causal effects are observed. An arrow implies that in at least all Markov equivalent graphs, this arrow is present. A causal effect is then obtained from the smallest estimated average causal effect, obtained using (\ref{eq:average-causal-effect}), across all possible graphs in the Markov equivalence class. Note that in the graph of causal effects (Figure \ref{fig:causal-effects-p1}) there are many feedback relations even though these are not visible in the causal graph (Figure \ref{fig:network-p1}).

The causal effects, here obtained from the FCI graph, are collected in the $p\times p$ kernel matrix $K_{\mathcal{G}}$. This kernel with causal information, however obtained, can be used to transform to a Markov process which may reveal interesting patterns over time. 

%---------------------------------------------------------------------------
\section{Non-linear dynamic Markov process}\label{sec:nonlinear-dynamic-markov-process}

To gain insight into possible trajectories, we use a non-linear Markov process which consists of an intrinsic (to a node) part that is bistable, and a coupling between the variables, induced by the causal effects previously obtained, that shifts the bistable process. 

\subsection{From causal effects to Markov process}

In the dynamics, each node is related to itself and nodes are coupled to other nodes by the causal effects, which is based on causal graph $\mathcal{G}$. 
In order to obtain a Markov process, the causal effects in kernel $K_\mathcal{G}$ obtained from the causal graph, are transformed into a Markov transition matrix. 
Kernel $K_\mathcal{G}$ can be related to a Markov process by \citep{Coifman:2006}
\begin{align}
    P_\mathcal{G}(i,j)=\frac{K_\mathcal{G}(i,j)}{d_\mathcal{G}(i)}
\end{align}
where $d_\mathcal{G}(i)$ is the degree (sum of weights) of $K_\mathcal{G}$ for row  (node) $i$ \citep{Seabrook:2023}. For an asymmetric kernel we note that the degree is then the out degree of the nodes in $\mathcal{G}$. In matrix notation we have
\begin{align}\label{eq:kernel-transition}
    P_\mathcal{G}=D^{-1}_\mathcal{G}K_\mathcal{G}
\end{align}
where $D_\mathcal{G}=\text{diag}(d_\mathcal{G}(1),d_\mathcal{G}(2),\ldots,d_\mathcal{G}(n))$ is a $p\times p$ diagonal matrix with the (out)degrees of $K_\mathcal{G}$. The matrix $P_\mathcal{G}$ can be interpreted as a matrix of transitions for a Markov process on variables $X_i$ with initial value $X_0$. For notational convenience, we often ignore the subscript $\mathcal{G}$ if it is clear which information we use to make the kernel. 

If it is assumed that the transitions remain fixed over time (time homogeneous), then the change from all values of the nodes $(X_{i,t},i\in V)$ to $X_{j,t+1}$ for one step in time from $t$ to $t+1$ is given by \citep{Ravazzi:2015,Levin:2017}
\begin{align}\label{eq:linear-markov-single}
    X_{j,t+1} = \sum_{i\in V} P(i,j)X_{i,t}
\end{align}
This represents the inputs of all variables $X_{i,t}$ for $i=1,2,\ldots,p$, including the node $j$ itself weighted by the transitions $P(i,j)$. 
In matrix notation for $p$-vector $X_{t}$ and $p\times p$ matrix of transitions $P=(P(i,j),i,j\in V)$ we obtain for $k$ steps in time \citep{Stroock:2005}
\begin{align}\label{eq:markov-dynamics}
    X_{t+k} = PX_{t+k-1}=P(PX_{t-k-2})=\cdots = P^kX_t
\end{align}
The $ij$th element of $P^k$ is denoted by $P_k(i,j)$ and represents the transition from node $i$ to $j$ in $k$ steps. Some relevant properties of such processes are briefly discussed in Appendix \ref{app:markov-chains}.

This type of evolution is sometimes referred to as diffusion, and is useful to determine long term behaviour like fixed points \citep{Ravazzi:2015} or communities \citep{Coifman:2006}. 

\subsection{Making the model non-linear}

Although such a linear process is already useful, it is often used for its convenient mathematical properties (e.g., convergence) rather than a linear model emulating realistic features of the process under investigation. Here we seek to emulate (some of) the properties of mental well and ill being more realistically, which  requires a non-linear model \citep[e.g.,][]{Hayes:1998}. Non-linear features in the context of mental well being that have been  mentioned in the literature as being important to take into account are, among others, bistability, sudden transitions, and hysteresis \citep[e.g.,][]{Hayes:1998,Leemput:2014,Kossakowski:2019b,Waldorp:2020,Scheffer:2024}. 

The general non-linear effect is inspired by network models of interacting elements \citep{Leemput:2014,Maas:2020,Kohler:2021,Wunderling:2021,Scheffer:2024}. The interacting elements induce stable states and transitions between these states, and tipping (critical) points that are unstable states such that transitions may occur with small perturbations \citep{Strogatz:1994,Scheffer:2001}. 

We choose a non-linear and one-to-one (bijective) mapping of the Markov process in (\ref{eq:markov-dynamics}) that accords with many of the properties like bistability, hysteresis, and boundedness \citep{Hayes:2020,Schiepek:2023,Hofmann:2025}: $f(x)=\tanh(x)=1/(1+\exp(-x))$, the hyperbolic tangent. In addition to the fact that the hyperbolic tangent is in line with many features of real life processes, it is often used in modelling non-linear processes in physics \citep[e.g.,][]{Kohler:2021,Pham:2024}, in neural networks \citep{Bishop:1995,Rojas:2013}, in machine learning/artificial intelligence \citep{Hastie:2001,Shalev:2014}, in social sciences \citep{agresti:1996,Maas:2020,Finnemann:2026}, and in clinical science \citep{Borsboom:2011,Cramer:2016,Loossens:2020}, making it a well-known model that appears to be appropriate in many fields.
The function $f$ is bounded by $-1$ and $1$, has two fixed points close to $-1$ and $1$, and an unstable fixed point at $0$ (for a given parameter range, see further down). We show $f$ in Figure \ref{fig:tanh}(a), where the intersection with the $45^\circ$ line defines the fixed points; the top and bottom fixed points are stable (blue) and the middle fixed point is unstable (red). A small perturbation to the system close to an unstable fixed point, may lead to a sudden transition \citep{Strogatz:1994,Scheffer:2024}. 
\begin{figure}[t]\centering
\begin{tabular}{@{\hspace{-2em}} c @{\hspace{-3em}} c @{\hspace{-1.5em}} c}
    
    \pgfimage[width=0.4\textwidth]{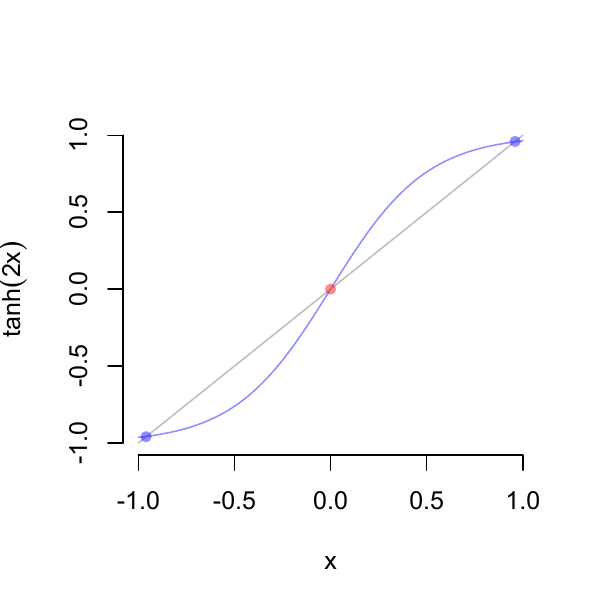}
    &
    \pgfimage[width=0.4\textwidth]{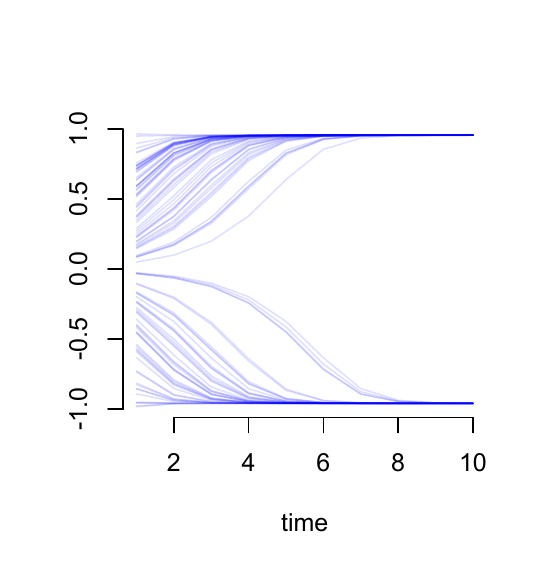}
    &
    \pgfimage[width=0.4\textwidth]{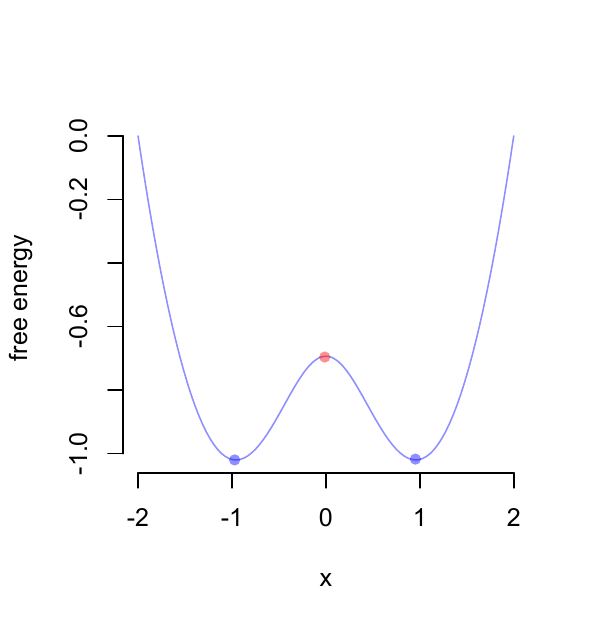}\\[-0.8em]
        (a) & (b)   & (c)\\[-3em]
    \pgfimage[width=0.37\textwidth]{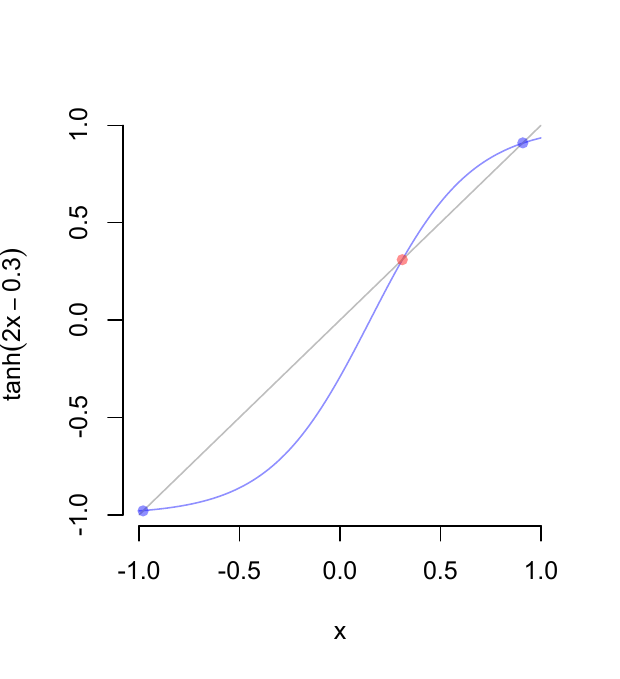}
    &
    \pgfimage[width=0.4\textwidth]{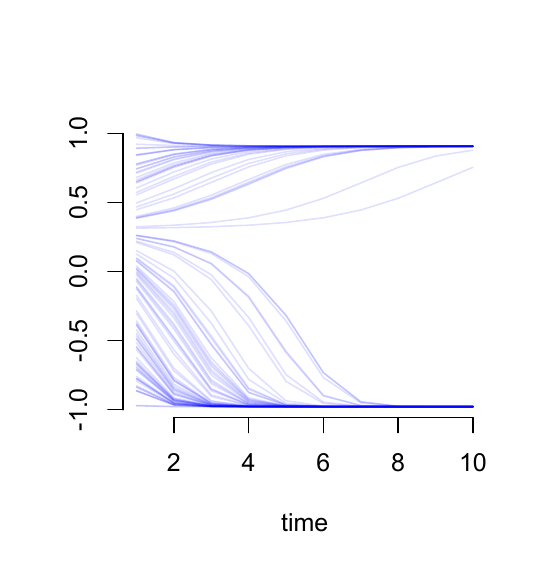}
    &
    \pgfimage[width=0.4\textwidth]{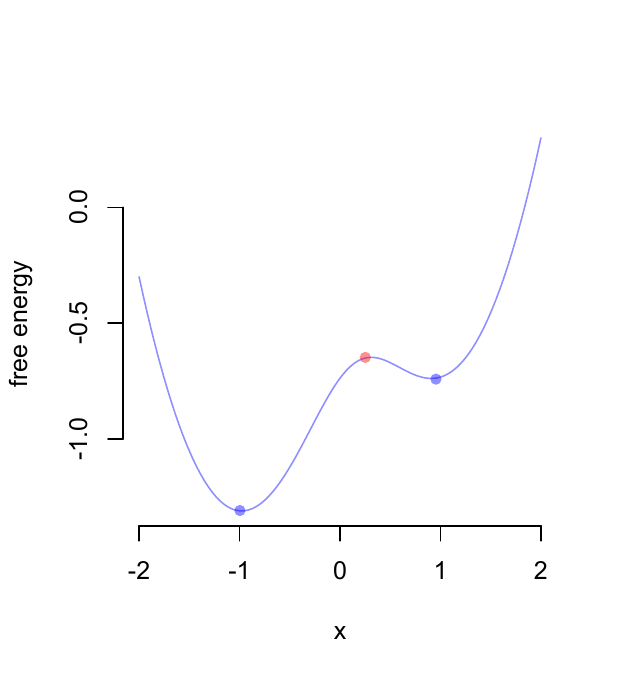}\\[-0.8em]
    (d) & (e) & (f)
\end{tabular}
\caption{ The hyperbolic tangent function in (a) $\tanh(2x)$; the intersection with the $45^\circ$ line gives the three fixed points, the upper and lower fixed points are stable (blue) and the middle fixed point is unstable (red). In (b) is the recursive function $x_{t+1}=\tanh(2x_{t})$, showing that the 100 trajectories at random initial values in $[-1,1]$ move to the stable fixed points close to $-1$ and close to $1$. In (c) we see the free energy of $\tanh(2x_t)$, with the same fixed points as in (a). In (d), (e) and (f) we see a similar set of figures of $x_{t+1}=\tanh(2x_t-0.3)$. }
\label{fig:tanh}
\end{figure}

The evolution of $x_{t+1}=\tanh(2x_t)$ is shown in Figure \ref{fig:tanh}(b), illustrating the process moving towards either one of the stable fixed points, determined by the initial value. In Figure \ref{fig:tanh}(c) is the free energy (the integral of the mean field, see Appendix \ref{app:analyses-R}), showing the two stable fixed points at the extremes, and the unstable fixed point in the center. The free energy representation shows that the two minima are stable, in that, like a ball, it will remain in (or near) the minima given small perturbations, while at the local maximum in the center, a ball will move away given a small perturbation. 

Changes in state can also be effectuated by changes to the shape of the free energy landscape \citep{Scheffer:2001,Cui:2025}. The parameter $\tau$ in $x_{t+1}=\tanh(\gamma x_t + \tau)$, where $\gamma$ referred to as the (intrinsic) strength and $\tau$ the external field \citep{Plischke:1994}, makes the free energy landscape asymmetrical. In Figure \ref{fig:tanh}(f) we see the free energy of $x_{t+1}=\tanh(2x_t-0.3)$, showing that $\tau=-0.3$ increases the proportion of trajectories that end up in the lower fixed point. The corresponding trajectories over time in Figure \ref{fig:tanh}(e) and the fixed points in (d), also show that the center with the unstable fixed point has shifted and the basin of attraction is now larger for the lower stable fixed point. Hence, we can distinguish between the intrinsic process of the variable affecting itself from one time point to the next, and the other variables affecting the likelihood of the trajectory to move either upward or downward. 

Because of the features of the hyperbolic tangent described above, we use it to make the linear Markov process in (\ref{eq:markov-dynamics}) non-linear 
\begin{align}\label{eq:non-linear-markov-process}
X_{j,t+1} =f(X_{j,t}) = \tanh\Bigl( P(j,j)X_{j,t} + \sum_{i\ne j} P(i,j)X_{i,t} \Bigr) 
\end{align}
where $\gamma=P(j,j)$ represents the internal part, the link of the variable to itself, and $\tau= \sum_{i\in V} P(i,j)X_{i,t}$ is the external field. In this formulation the coupling of the other variables represents the external field \citep{Maas:2020,Kohler:2021}, which may change the symmetry of the free energy landscape shown in Figure \ref{fig:tanh}(f). Obviously, other factors (e.g., background variables) could also be included in the external field. Noise is included in the simulations by including the inverse temperature $\beta=1/T$, multiplied with all terms, so that $x_{t+1}=\tanh(\beta (\gamma x+\tau))$; the effect is similar in distribution to the cubic approximation and adding a Gaussian noise variable \citep[see][]{Cobb:1980}. In order to achieve bistability, we set  $\beta = 2$ throughout (see also Appendix \ref{app:analyses-R}).
%It is easily seen that (\ref{eq:non-linear-markov-process}) is a non-linear Markov process, since it depends only on $X_{i,t}$ for $i\in V$. 
In Appendix \ref{app:asymptotic-properties-linear-Markov} we show that this map converges, even if the transitions are allowed to change over time, given certain conditions.

\subsection{Example simulation}

We simulated this process 200 times using data of a patient with a mental disorder to obtain the causal effects for transitions $P(i,j)$ (given in Figure \ref{fig:causal-effects-p1}) and the hyperbolic tangent in (\ref{eq:non-linear-markov-process}) for random initial values in the interval $[-1,1]$. We include many random initial conditions because we account for the uncertainty in which state a patient is at the start of treatment. It is, of course, also possible to use a smaller interval around the last time point known of the patient. Figure \ref{fig:example-patient2} shows the result of these simulations for the eight variables. We have chosen 16 time points here to roughly correspond to 16 weeks of treatment provided once a week. The correspondence between the time points used here and actual time in practice is not clear cut; we can interpret time in the simulation in terms of time points at which the treatment is given, which we here take as once a week. 
\begin{figure}[t]\centering
    \pgfimage[width=\textwidth]{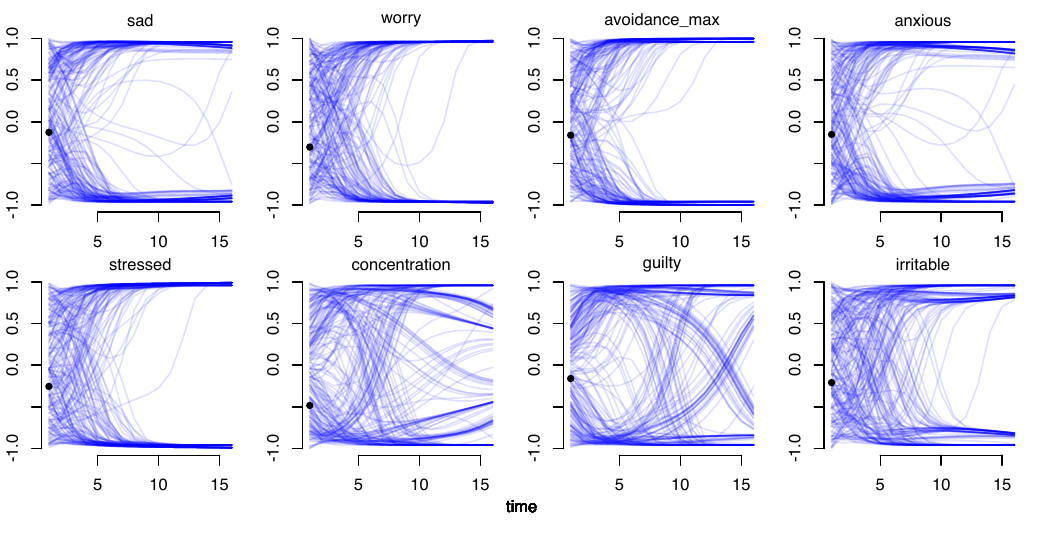}
\caption{ Simulation of 200 trajectories of the non-linear Markov process in (\ref{eq:non-linear-markov-process}) with $\gamma=2$ and the causal effects as in Figure \ref{fig:network-p1}, with initial value chosen randomly from $[-1,1]$ for each run independently. The black dots at $t=0$ are the values of the last time point for this patient $x_0=(-0.13, -0.30, -0.16, -0.15, -0.25, -0.48, -0.16, -0.21)$.  }
\label{fig:example-patient2}
\end{figure}
We see in Figure \ref{fig:example-patient2})that a trajectory is not independent of the initial condition, but also strongly depends on the effect of the other variables, the coupling from the causal effects; this can be seen by many of the trajectories that start at values $>0$ end up at the lower fixed point and vice versa. 

The next step is to determine what could happen given a particular simulated intervention and compare that to a control (not intervening, as shown in Figure \ref{fig:example-patient2}); we therefore need to implement different types of intervention. 

%---------------------------------------------------------------------------
\section{Simulating and comparing interventions}\label{sec:interventions-comparisons}

Having a model to simulate the effects of an intervention, we here describe several possible interventions (by no means exhaustively), and then provide some guidelines on how to evaluate the success of an intervention. Currently, we are mainly focusing on CBT, where one or some of the problems (e.g., symptoms) are trying to be reduced by exposure, for example, or where, together with the patient, a goal is set to be achieved, e.g., engage more in social activities to reduce social anxiety \citep{Santen:2024}. Here we try to capture these interventions by changing something in the network and measuring the effects of these changes using simulations. 

%---------------------------------------------------------------------------
\subsection{Interventions on networks}\label{sec:interventions-networks}
An intervention is an externally induced change to the network. Many different interventions on networks are possible. In terms of the network, the most obvious intervention is on the nodes \citep{Valente:2012,Zhang:2018,Ryan:2025b}, but strategies to affect the system by the external field \citep{Cramer:2016,Jimenez:2023,Cui:2023} have also been suggested. We discuss several here, but this is not an exhaustive list. 
We consider three different interventions
\begin{itemize}
    \item[($i$)] perturbing the level (value) of a node
    \item[($ii$)] (transient) driving towards a target (i.e., adaptive control)
    \item[($iii$)] changing the connection between two nodes (influence)
\end{itemize}
We explain each of these interventions in turn.
%Many more variations of interventions exist, although most are geared towards social behaviour \citep[see][for an overview]{Valente:2012}, including combinations of the above; combining ($ii$) and ($iii$), for example.

\subsubsection{($i$) Intervening on a node}

A change (perturbation) of the level of a node can be represented as a change in initial condition \citep[sometimes called pulse,][]{Ryan:2025b}, i.e., at time $t=0$ the value of node $i$ is set to $X_{i,0}=x_{i,0}$ and at time $t=1$ the process is no longer fixed for variable $X_{i,1}$. Then the dynamics are obtained from (\ref{eq:non-linear-markov-process}). In contrast to a linear process where such an intervention has no long term effect \citep[see e.g., ][and Lemma \ref{lem:asymptotically-uncorrelated} in Appendix \ref{app:asymptotic-properties-linear-Markov} with $U_t=0$ almost surely]{Levin:2017,Ryan:2025b}, in the non-linear Markov process $f(x)=\tanh(x)$, a small shift can lead to a different stable state. For example, consider Figure \ref{fig:tanh}(b), where the initial condition $>0$ determines that the trajectory will move towards the upper stable state, and the other way around for initial condition $<0$. Hence, when a person is near $0$, near the middle unstable fixed point (see Figure \ref{fig:tanh}(a)), a small change in the initial value (perturbation) may change the long term behaviour.

\subsubsection{($ii$) Driving node towards target}

Here, the objective is to reach a particular goal, for example, to reduce avoidance to a lower level. This could be an example of goal directed therapy in cognitive behavioural therapy. In a network we can achieve this by setting a target $u$ such that it represents this reduction, but still depends on the other variables. This can be achieved by adding in the linear model (\ref{eq:linear-markov-single}) a target, such that we obtain \citep{Chow:1975}
\begin{align}\label{eq:tanh-control-1}
X_{j,t+1}=P(j,j)X_{j,t} + \sum_{i\ne j}P(i,j)X_{j,t}+u_j.
\end{align}
Such an intervention corresponds to a soft intervention, where the existing causal effects remain but there is an additional (external) effect $u$ \citep[see Definition \ref{def:soft-intervention} in Appendix \ref{app:assumptions-causal-discovery},][]{Eberhardt:2007}. In the example of stress reduction, we set $u_j=-0.3$, where $j$ refers to the node \textit{stressed}. Then we apply control theory  \citep[e.g.,][]{Chow:1975,Brunton:2022} to reach this target, which we explain in more detail below. 

More realistically, this type of control changes over time, either because the patient is less inclined to be occupied with therapy or because something changes in the treatment. Therefore, it could also be the case that the target is transient. For example, there is a surge at the start of the therapy to reach the target but this is reduced after two or three weeks. We then have that the target depends on time, $u_t$. To emulate such a transient process we use the gamma density function $\Gamma(x,s,r)=x^{s-1}e^{rx}r^s/\Gamma(s)$, with shape parameter $s$ and rate parameter $r$ and $\Gamma(s)$ the gamma function \citep{Bickel:2007}. Three versions of this gamma density are shown in Figure \ref{fig:gamma}. 
\begin{figure}\centering
\begin{tabular}{c @{\hspace{-1em}} c @{\hspace{-1em}} c}
    \pgfimage[width=0.35\textwidth]{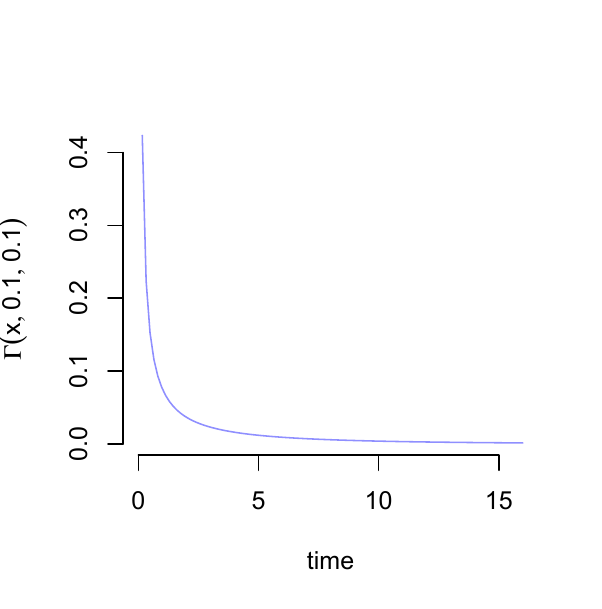}
    &
    \pgfimage[width=0.35\textwidth]{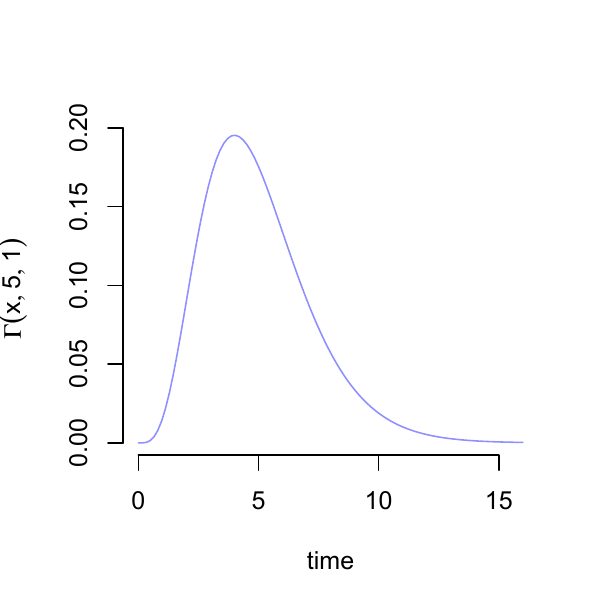}
    &
    \pgfimage[width=0.35\textwidth]{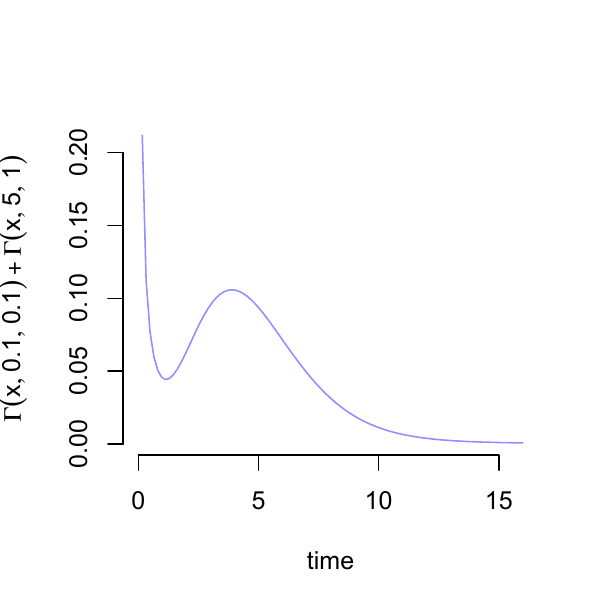}\\
    (a) & (b) &(c)
\end{tabular}
\caption{ The gamma density function $\Gamma(x,s,r)$ shown in (a) with $s=0.1$ and $r=0.1$, in (b) with $s=5$ and $r=1$, and in (c) with  a mixture $0.5\Gamma(x,0.1,0.1) + 0.5\Gamma(x,5,1)$. }
\label{fig:gamma}
\end{figure}

We use a time-dependent target $u_t$ that changes intensity over time and affects the process while the original coupling effects remain (soft intervention, see Appendix \ref{app:causal-basics}). This results in the target $u_t$ being additive and affected by the transitions in $P_k(i,j)$ (see the remark in Appendix \ref{app:asymptotic-properties-linear-Markov}). The process in (\ref{eq:tanh-control-1}) after $t+k$ steps, starting at $t$ is then 
\begin{align}\label{eq:tanh-control}
    X_{j,t+k} 
    % = f \left( P^kX_t + \sum_{s=0}^{k-1}P^s u \right)
        = f \Bigl( P_k(j,j)X_{j,t} + \sum_{i\ne j} P_k(i,j)X_{i,t} + \sum_{i\in V}C_k(i,j) u_{i,t} \Bigr)
\end{align}
where $P_k(i,j)=(P^k)_{ij}$ and $C_k(i,j)=[(I-P^k)(I-P)^{-1}]_{ij}$, and $P$ is the transition matrix; this is the result of the target $u_{i,t}$ considered as external and additive \citep[see, e.g.,][Chapter 2]{Galor:2007}. In control theory, the target $u_t$ can also depend on the states $x_{i,t}$. This may be realised by a sequence of interventions that guides an individual to a target by changing the intervention now and then \citep{Zabczyk:2009}. We show in Appendix \ref{app:asymptotic-properties-linear-Markov} under which conditions this map converges, also for random target $U_t$. 

We illustrate the effect of a transient intervention on the trajectories obtained from the causal effects of real data and mapping them using (\ref{eq:tanh-control})  in Figure \ref{fig:intervention-stress}. The intervention is on the single variable \textit{avoidance}, and is transient with $s=5$ and $r=1$, as in Figure \ref{fig:gamma}(b). 
\begin{figure}\centering
    \pgfimage[width=\textwidth]{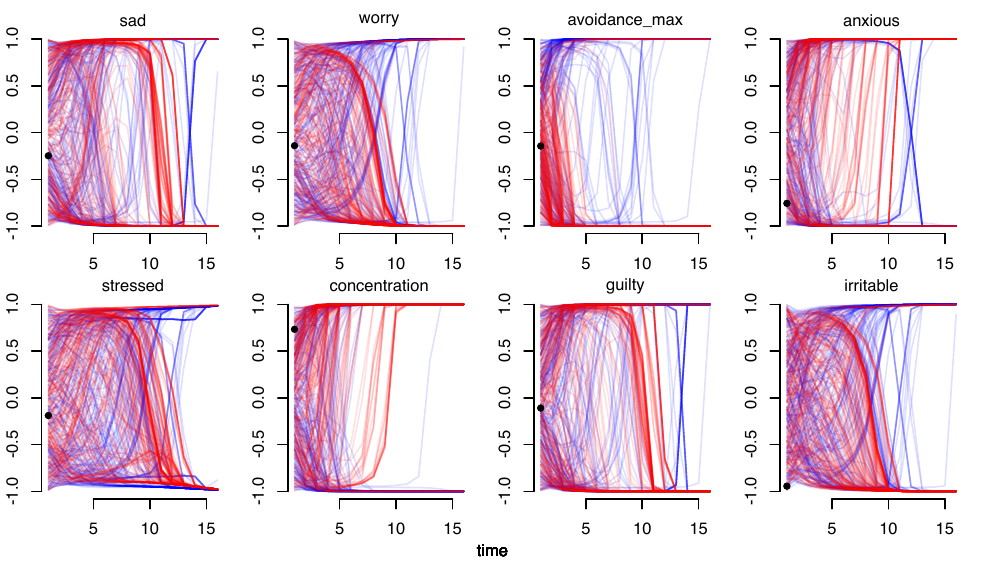}
\caption{ Simulations of 200 trajectories for both control (no intervention, blue) and intervention cases (red). Initial values are drawn randomly from the interval $[-1,1]$ and the dot represents the final position of the patient in the EMA data. A transient intervention on the single variable \textit{avoidance} is applied, with parameters  $s=5$ and $r=1$ (red trajectories). }
\label{fig:intervention-stress}
\end{figure}
It is clearly seen that the transient intervention has the desired effect since almost all intervention trajectories (red lines in Figure \ref{fig:intervention-stress}) improve (except \textit{anxious}). This shows that the coupling really matters in that only the variable \textit{avoidance} is intervened upon, while all other variables improve as well. This also connects to interventions on feedback relations; in Figure \ref{fig:causal-effects-p1} \textit{avoidance} is part of two cycles and the intervention has large effects (see also Section \ref{sec:comparing-interventions}). Obviously, it is also possible to intervene on two or more variables simultaneously; in Figure \ref{fig:intervention-35} in Appendix \ref{app:analyses-R} we show an example of an intervention on both \textit{concentration} and \textit{avoidance}. It shows that the reduction in symptom levels decreases even faster.

Corresponding to the illustration in Figure \ref{fig:intervention-stress}, we show the free energy in Figure \ref{fig:free-energy-over-time}, using the mean field approximation (see Appendix \ref{app:analyses-R}), similar to Figure \ref{fig:tanh}(a), of the variable sad. Each plot in Figure \ref{fig:free-energy-over-time} represents a landscape obtained
from the free energy averaged over the simulated data for the control and intervention condition, based on the real data for the causal effects. We can therefore see how the landscape changes as it unfolds over time. We see in Figure \ref{fig:free-energy-over-time} that over time the landscape of the free energy and the minima (attractors), where the variable sad will end up, changes. For the control condition (blue) the minima become more equal (symmetric) over time, while for the intervention condition (red), the minima first become more balanced and then change more toward $-1$ and the minimum at 1 has disappeared. This illustrates the effect of the intervention with respect to the control in relation to the free energy and how this change
might be thought of as a healthier state. Note that this change is because of both the change in the variable sad itself and the coupling with the other variables. Connecting to sad in Figure \ref{fig:intervention-stress} (top left), at earlier times some of the trajectories move upwards, but at later times, nearly all trajectories move downwards,
aligning with the reshaping of the free energy landscape in Figure 8.
\begin{figure}\centering
    \pgfimage[width=\textwidth]{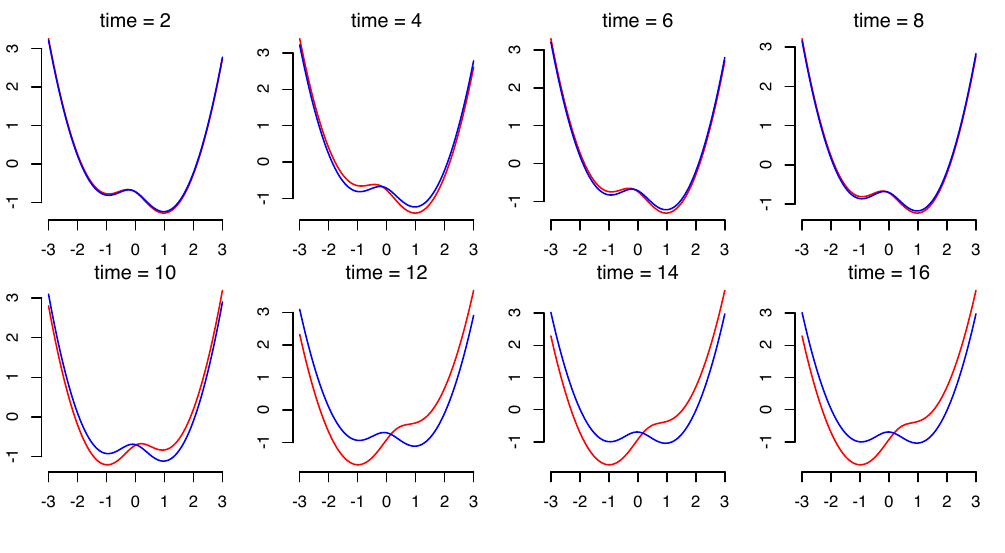}
\caption{ Free energy approximation of the mean field for the variable sad in both the con-
trol (blue) and intervention (red) condition for different snapshots over time. In the control
condition the minima balence out over time, while in the intervention condition the global
minimum changes toward the lower (improved) end. }
\label{fig:free-energy-over-time}
\end{figure}

\subsubsection{($iii$) Intervening on connections between nodes}

Changing a connection between two nodes represents the process of removing an association, or strengthening the influence of one variable on another node. A change in connection is performed at the level of the kernel $K_\mathcal{G}$, which is then transformed to a transition matrix $P_\mathcal{G}$  (see (\ref{eq:kernel-transition})). If we change $K_\mathcal{G}(i,j)$ to $0$, for example, then we can consider the effect of severing the connection $i-j$. We can then determine the dynamics using (\ref{eq:non-linear-markov-process}) and compare the effect with the control (i.e., no intervention), similar as before. Making a change in the kernel $K_\mathcal{G}$ and then transforming to $P_\mathcal{G}$ ensures that the properties to obtain convergence still hold. However, this process possibly also changes (reduces or enhances) some of the other causal effects. Alternatively, a change can also be made directly to the transition matrix $P_\mathcal{G}$ so that nothing else is effected. Convergence is then not guaranteed, although it could still be valid, but the induced change is more in line with the conceptual representation.

%------------------------------

\subsection{Comparing starting points of interventions}\label{sec:comparing-interventions}

To compare the effectiveness of the treatment, we use metrics that are commonly defined for linear Markov processes on nodes \citep{Coifman:2006} and on connections \citep{Hammond:2013}. Similar to statistical physics \citep{Plischke:1994}, we consider the average value $\bar{x}_t$ (over all variables) for both the control and intervention trajectories as an order parameter to use for comparisons. This works well because the values are bounded to the interval $[-1,1]$ and a value below or above 0 indicates how well the treatment is going. Hence, comparing the difference as an effect measure seems reasonable. This idea is inspired by linear Markov processes \citep[][see Appendix \ref{app:distance-measures-linear-markov}]{Coifman:2006}. 

To show how this works, and at the same time consider which intervention might be most efficient, we consider an example. Using the same patient data with causal graph and effects in Figure \ref{fig:causal-effects-p1} and such effects as shown in Figure \ref{fig:intervention-stress}, we can compare the effect of perturbing a single node on all other nodes and evaluate the difference between the control and intervention trajectories on average over time. This is shown in Figure \ref{fig:intervention-effect-time}, where effect size is given for the eight variables, and so a high value means that the effect of the intervention is large. 
\begin{figure}[t]\centering
    \pgfimage[width=\textwidth]{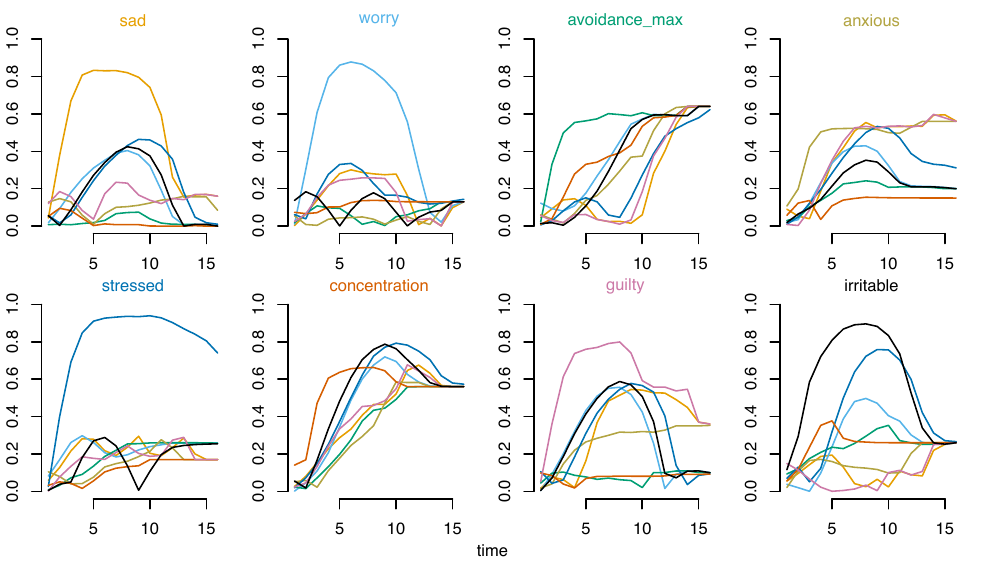}
\caption{ Effects over time comparing control and intervention trajectories on average using 200 simulations each.  The color of the lines in each figure correspond to the color of the variable name in the title of each figure. The settings for the intervention are the same as for Figure \ref{fig:intervention-stress}. Variable names in the title are the intervention variable in each figure.  }
\label{fig:intervention-effect-time}
\end{figure}
Each line can be interpreted as the average difference between all of the control and intervention trajectories at each time point for an intervention of the variable name in the title. This explains why the colour of the fastest growing line in each figure corresponds to the name of the variable. From such a plot we see that, for example, effects of perturbing \textit{stressed} is that stress improves (large effect) but the effect on all other variables remains low. This suggests that the effect of an intervention on \textit{stressed} does not propagate through the network and so does not lead to overall improvement. In contrast, perturbing \textit{avoidance} or \textit{anxious} shows a a larger effect on most other variables. Additionally, the effect remains and starts early after the onset of the intervention. Hence, it would make sense to try and intervene on either of these variables or both (see Figure \ref{fig:intervention-35} in Appendix \ref{app:analyses-R}), if possible. Such evaluations take into account the fact that there is a sizeable effect, that the effect takes hold quickly, and that the effect remains for a longer time. 

In an attempt to quantify such results, we use two different measures based on the linear case (see Appendix \ref{app:distance-measures-linear-markov}). The results are shown in Table \ref{tab:intervention-results}. First, for both the control and intervention condition, we take the average across all simulations at the last time point (in our example 16). These are in the first two columns in Table \ref{tab:intervention-results}. Notice that the averages for the control condition are almost all near 0, while for the intervention condition several are $\neq 0$. Then we consider the difference between control and intervention; this is in the third column in Table \ref{tab:intervention-results}. The difference refers to the effect the intervention has with respect to the control on all variables simultaneously at the last time point. We see that the effects for \textit{avoidance} and \textit{concentration} are strongest, implying that by intervening on either of these, the effect on all other variables through the coupling is relatively strong. Note that the effect of \textit{anxious} is also reasonably strong, but that it is in the wrong direction. Drawbacks of this measure are that (a) it only represents the last time point and does not take into account the swiftness of the response (steepness of the improvement), and (b) it is bounded by $-2$ and $2$, and so may have difficulty in distinguishing between variables. Therefore, we also use a second measure: we take the averaged effects of both control and intervention condition across the entire time interval, then sum the differences between the control and intervention averages, and then we average across all variables; this is in the final column called \textit{sum effect}. 
\begin{table}
\begin{tabular}{ l r r r r}
perturbation     	&control		&treat 	&control - treat 		&sum effect\\ \midrule
sad            	&-0.02 		&-0.08      &0.06  			&3.61\\
worry         	&-0.01 		&-0.03      &0.02 			&2.53\\
avoidance  	&0.02 		&-0.31      &0.33 			&5.49\\
anxious        	&0.01	 	&0.23      &-0.22 			&4.33\\
stressed      	&0.01 		&-0.06      &0.07			&2.30\\
concentration	&0.02		&0.27     	&-0.25 			&5.05\\
guilty        		&0.03 		&-0.07      &0.10			&3.80\\
irritable      	&0.01 		&0.02      &-0.01 			&2.68\\ \midrule
\end{tabular}
\caption{Effects of the averaged difference between control and intervention trajectories over the entire time interval. }
\label{tab:intervention-results}
\end{table}
Similar to the results from the plots and on the first measure, we see that the largest effects are on \textit{avoidance} and \textit{concentration}. Referring back to Figure \ref{fig:causal-effects-p1}, both variables  \textit{avoidance} and \textit{concentration} are part of cycles with negative feedback relations, and so it seems reasonable to expect large effects. These results are available for five patients, including patient 1 discussed here, in Appendix \ref{app:results-5patients}. The five cases illustrate different treatment foci following from the framework; patient 1 \textit{avoidance} and \textit{concentration}, patient 2 \textit{avoidance}, patient 3 \textit{stress} and \textit{concentration}, patients 4 and 5: \textit{concentration}.

%---------------------------------------------------------------------------
\section{Discussion and conclusion}\label{sec:discussion-conclusion}

Clinical decision-making regarding the choice of treatment in mental health care is generally based on a combination of diagnostic classification (for example, according to the DSM-5), generic models tailored to the individual patient, and individual case conceptualization based on clinical expertise, intuition and experience. This approach involves many instances of subjective decision-making, with the risk of numerous biases. To address these limitations, data-driven and process-oriented approaches to clinical case formulation may improve clinical decision making and the associated choice of treatment strategies. Here we have proposed a comprehensive (covering methods for each stage of the process) and flexible (each component can be replaced or adapted) framework to support clinical decision-making. It is specifically designed to provide recommendations for selecting the treatment focus prior to the start of treatment, taking full account of the patient’s specific symptoms, feelings, behaviour, beliefs, and contextual factors over time. 

The framework outlined here consists of three steps: (a) in the first step, a causal graph is made with causal effects based on EMA data, then (b) in step two, various scenarios are assessed, based on a non-linear model in which either no action is taken (control) or an intervention (a perturbation) is simulated, and, finally, (c) step three determines the extent to which the intervention is more successful than the control without intervention. The framework clearly indicates the starting point of treatment in order to be as effective as possible. This process can, of course, be repeated throughout the course of treatment, given that the patient’s condition—and therefore both the causal graph and the causal effects—change as treatment progresses. 
The current framework is ideally suited for use with EMA, which provide a wealth of information measured in everyday life outside of treatment sessions; the data can be collected before the actual treatment begins, and assessments can continue throughout the course of treatment. EMAs reflect the patient’s daily experiences before they are distorted by the patient’s memory bias or the therapist’s interpretation bias. Also of interest is that EMAs can capture disorder specific and comorbid symptoms within one causal graph, making it a truly personalized approach based on the patient’s symptoms and not constrained by diagnostic boundaries. 
If EMAs are not available, other causal graphs can be used within this framework, such as generic cognitive behavioural models (CBM) tailored to the individual patient, or perceived causal networks \citep[PECAN,][]{Andreasson:2023,Klintwall:2023}. Both CBM and PECAN construct causal networks of variables considered relevant by the patient and the therapist. They can also be combined with EMA \citep[e.g.,][]{Scholten:2025}. Future research could examine whether incorporating perceived relationships into empirically driven EMA causal graphs adds value, or whether it might even detract from it.

The model could also be further developed. We now have quantified success based on the extent of the difference between intervention and control trajectories and responsiveness to the intervention. Another option is to take the strongest response to the intervention (averaged across all variables). Several approaches are possible, but only a few will prove to be empirically relevant. We believe that our framework is fully modular in this respect.

The data are directly related to the causal effects, which are used to simulate the time series with the non-linear model. We chose a model that matches generic phenomena observed in data, but the non-linear model was not fitted to data. This can be improved by, for example, estimating which Taylor expansion will fit the data best using polynomial regression \citep[e.g.,][]{Brunton:2016,Brunton:2022}. Also, the measures of causal effects could be replaced by, for example, measures of relative entropy \citep{Janzing:2013}. Depending on the results of future testing, the model could be changed to equally valid or more suitable models, such as the Blume-Capel mean field model with three or more stable states \citep[][]{Maas:2026,Waldorp:2026}. 
Finally, it is of great importance investigating the clinical validity of our framework in the processes of case formulation and choice of treatment focus in future research. The validation of the current framework should demonstrate that it makes a valuable contribution to clinical decision-making – decisions that would not be made without this framework, and that result in more successful treatment.

%-------------------
\paragraph{Acknowledgements.} JMBH was supported by the Netherlands Organisation for Scientific Research (NWO) under VENI grant number 221G.110. JMBH, LJW AJ and TM have been supported by the gravitation project ‘New Science of Mental Disorders’ (www.nsmd.eu), supported by the Dutch Research Council and the Dutch Ministry of Education, Culture and Science (NWO gravitation grant number 024.004.016) awarded to AJ. 

\paragraph{Materials.}

%---------------------------------------------------------------------------
%\bibliography{references}

%---------------------------------------------------------------------------
\newpage
\appendix
%---------------------------------------------------------------------------

%----------------------
\section{Causal inference basics}\label{app:causal-basics}
Let $\mathcal{G}$ be a graph with nodes (vertices) $V=\{1,2,\ldots,p\}$ and edges $E$ between the nodes. With each node $i$ we identify a random variable $X_i$ with marginal distribution $\mathbb{P}(X_i)$ and joint distribution $\mathbb{P}(X)$. A vector of variables for nodes in subset $A$ is denoted $X_A=(X_i, i\in A)$. We introduce informally three different types of edges: an undirected edge $i-j$ means that nodes $i$ and $j$ are adjacent, a directed edge $i\to j$ means that $i$ causes $j$, and a bidirected edge $i\leftrightarrow j$ means that there is uncertainty about the causal relation and there might be a confound. If $i\to j$ we also say that $i$ is a parent of $j$ and $i\in \textsc{pa}(j)$. A path between subsets $A$ and $B$ of $V$ is a set of edges between a node $i\in A$ and a node $j\in B$ such that $i-\cdots - j$, for any set of edges. A directed path from $A$ to $B$ is the set of edges such that $i\to \cdots \to j$. A collider is node $k$ on the path $i\to k\leftarrow j$. 

We start with an undirected graph where all edges are of the type $i- j$ to explain the concept of separation between disjoint subsets of nodes $A$, $B$ and $S$ from $V$ such that $A\cup B\cup S=V$. $A$ and $B$ are separated by $S$ in $V$ if for all paths between $A$ and $B$ there are nodes in $S$, i.e., $S$ is a cutset \citep[][]{Lauritzen96,Wainwright:2019}. This concept of separation is connected to conditional probability for the distribution $\mathbb{P}$ over the vector $X$ in $\mathbb{R}^p$ in the following way. If $A$ and $B$ are separated by $S$ in graph $\mathcal{G}$, then the variables $X_A$ are conditionally independent of $X_B$ given $X_S$ in distribution $\mathbb{P}$. This is called the Markov property, and sometimes we say that the graph is (Markov) compatible with the distribution \citep{Pearl:2000}. The Markov property implies that, whenever we find that two variables $X_i$ and $X_j$ are conditionally dependent, then there is an edge $i-j$. The Markov property also implies that a distribution can be decomposed into a product of smaller units (cliques, fully connected subsets of nodes) that are easier to compute \citep[Hammersley-Clifford Theorem,][Chapter 11, see also Lemma \ref{lem:hammersley-clifford} in Appendix \ref{app:assumptions-causal-discovery}]{Wainwright:2019}. To infer that there is no edge $i-j$, we must also assume the implication the other way around, that is, whenever $X_i$ and $X_j$ are conditionally independent, then there is no edge $i-j$. This is called faithfulness \citep{Peters:2017}. 

In a directed graph, the concept of separation is extended in the following way. Let $\mathcal{G}^m$ be the moral graph of a directed graph $\mathcal{G}$ such that $\mathcal{G}^m$ has the same edges as $\mathcal{G}$, the arrow heads are removed, and whenever $i\to k\leftarrow j$, an additional edge $i-j$ is introduced. The separation property of the moral graph can then be extended to the directed graph in functions that only depend on $X_i$ and the parents of $i$, i.e., $\textsc{pa}(i)$ \citep[][Lemma 3.21]{Lauritzen96}. Then for disjoint subsets $A$, $B$ and $S$, $A$ and $B$ are $d$-separated by $S$ if and only if in the moral sub-graph $\mathcal{G}^m(A,B,S)$, including all edges pointing towards $A$, $B$ or $S$, $A$ and $B$ are separated by $S$ \citep[this equals the definition in Appendix \ref{app:assumptions-causal-discovery}, see][Chapter 3]{Lauritzen96}. 

Unfortunately, different directed acyclic graphs can have the same set of $d$-separations and, hence, the same set of conditional independencies. Two graphs are called Markov equivalent if their undirected edges (skeletons) are the same and if they have the same set of colliders \citep{Verma:1991,Andersson:1997}. For instance, $i\to j\to k$ is Markov equivalent to $i\leftarrow j \leftarrow k$. Markov equivalent sets grow exponentially with the number of nodes when it may be assumed that there are latent confounders \citep[embedded causal models,][]{Verma:1991}.

\subsection{Causal Identification/Discovery Strategies}

One way to solve this issue was developed by \citet{Spirtes:1996}, and resulted in the algorithm fast causal inference (FCI). Instead of assuming that all relevant variables are observed (causal sufficiency), it is assumed that some nodes remain unobserved, so that ($i$) $i\to j$ means that $i$ is a (possibly indirect) cause of $j$, and ($ii$) $i\leftrightarrow j$ means that in some of the possible graph in the Markov equivalence class, the edges are not consistent so that there is uncertainty, possibly also a confounder $c$ such that $i\leftarrow c \to j$ \citep[see e.g.,][]{Spirtes:2000,Zhang:2008c,Mooij:2020a,Waldorp:2025}. The result of FCI is therefore an equivalence class in which there is uncertainty about some of the causal relations. We provide more rigorous definitions and assumptions of FCI in Appendix \ref{app:assumptions-causal-discovery}. 

To obtain a causal graph $\mathcal{G}$ we use the algorithm fast causal inference \citep[FCI,][]{Spirtes:1996}, implemented in the \texttt{pcalg} package in \texttt{R}. The FCI algorithm for Gaussian variables, requires the covariance matrix of the data to determine partial correlations \citep{Spirtes:1996}. Since, we use individual time series, we can invoke knowledge of time-series analysis \citep{Brockwell:2009}, where assumptions like stationarity are important \citep{Ryan:2025}. A common choice in modelling time series is to use a lagged version \citep{Haslbeck:2021}, which we can use for the estimation of causal effects, that is $X_t = \phi X_{t-1} + \varepsilon_t$, where $\phi$ is a (set of) parameter(s) and $\varepsilon_t$ is independent Gaussian with mean 0 and variance $\sigma^2$ \citep[an autoregressive model with lag 1,][]{Brockwell:2009}. Then we can use the covariance matrix 
\begin{align}\label{eq:covariance-time-series}
C_X(1) = X^\top(I-\phi B)^\top (I-\phi B) X,
\end{align}
where $X=(X_i,i\in V)$ is a $T\times p$ matrix with all $T$ timepoints and $p$ variables, and $B$ is the backshift operator (i.e., shifting each observation by one lag). Although the choice of a lag one model is common, there are many considerations and it is well worth investigating if the model is appropriate \citep{Haslbeck:2025c}.

%Basically, FCI starts with a full graph, where any node is connected to any other node, and then removes an edge if there is evidence of conditional independence. For example, if the variables are assumed to be multivariate Gaussian, then conditional independence is tested using partial correlations. Such tests are performed for many combinations of subsets of the nodes $V\backslash \{i,j\}$, where $i,j$ are the nodes for which dependence is investigated. To obtain some of the directions, note that if the pattern $i-k-j$ is obtained for three nodes $i, j, k$, and conditioning on $X_k$ induces a correlation between $X_i$ and $X_j$, then the only explanation for this is $i\to k\leftarrow j$ \citep{Pearl:2000}. FCI is the first algorithm that properly takes into account the presence of latent confounds, such that the result is a mixed set of edges for which sufficient evidence of a causal is obtained, and a set of edges where possibly an unobserved confound may also explain the conditional independence results \citep{Spirtes:1996}. The concept of $d$-separation and the assumptions of FCI are given in Appendix \ref{app:assumptions-causal-discovery}.

%---------------------------------------------------------------------------
\section{Assumptions of causal discovery}\label{app:assumptions-causal-discovery}
A graph $\mathcal{G}$ is composed of a set of vertices (nodes) $V=\{1,2,\ldots,p\}$ and a set of edges (connections) $E$. Nodes are indicated by small case letters $i$, $j$, $k$; the associated random variables are denoted $X_i$, $X_j$ and $X_k$, and for a subset $A$ of $V$, we denote by $X_A$ the vector $(X_i:i\in A)$. If $i\to j$ then $i$ is a \textit{direct cause} (parent) of $j$, $i\in \textsc{pa}(j)$, and $j$ is a \textit{direct effect} (child) of $i$. For $i$ to be a cause of $j$, it means that setting $X_i$ to a specific value (say $x_i$), is followed by a change in the distribution of $X_j$ (see also below). A \textit{path} is a sequence of alternating nodes and edges such that each node is the direct cause or direct effect of the next node in the sequence, irrespective of the orientation of the edge (e.g., $i\to j\leftarrow \cdots \to k$). If on a path we have $i\to k \leftarrow j$, then $k$ is a \textit{collider} on this path. A path is a \textit{directed path} if all edges are in the same direction, i.e., $i\to j\to \cdots \to k$ is a directed path from $i$ to $k$.
%, and $i\leftarrow j\leftarrow \cdots \leftarrow k$ is a directed path from $k$ to $i$. 
If there is a directed path $i\to j\to \cdots\to k$, then $i$ is a \textit{cause} (ancestor) of $k$ and $k$ is an \textit{effect} (descendant) of $i$. In a \textit{directed acyclic graph} (DAG) there are no directed paths from a node to itself (no cycles). The separation property of nodes in a DAG is described by $d$-separation \citep[e.g.,][]{Spirtes:1993,Lauritzen96,Pearl:2000}.

\begin{definition}{\rm ($d$-separation)}\label{def:d-separation}
A path $(\ldots,i,k,j,\ldots)$ is blocked by set $S\subset V$ if 
\begin{itemize}
\item[$\circ$] there is a collider $i\to k\leftarrow j$ with neither $k$ nor any descendants of $k$ in $S$, or if

\item[$\circ$] there is a non-collider of the type $i\to k\to j$, $i\leftarrow k\leftarrow j$, $i\leftarrow k\to j$ such that $k\in S$. 
\end{itemize}
If all paths between subsets $A$ and $B$ are blocked by $S$, then $A$ and $B$ are $d$-separated by $S$.
%We denote that $A$ is $d$-separated from $B$ by $C$ as $A\independent B\mid C$.  Let $\bf{A}$, $\bf{B}$ and $\bf{C}$ be disjoint subsets of $\bf{V}$ in the graph $G$. If all paths between any node in $\bf{A}$ and any node in $\bf{B}$ are $d$-separated by $\bf{C}$, then $\bf{A}$ and $\bf{B}$ are $d$-separated by $\bf{C}$ and denote this by $\bf{A}\independent \bf{B}\mid \bf{C}$.
\end{definition}
\begin{definition}{\rm (Causal evidence, \citealp[][Book III, Chapter XII]{Mill:1843}; \citealp[][Chapter 9]{Spirtes:1993}; 
\citealp[][Proposition 10]{Mooij:2020a})}\label{def:causal-evidence}
Let $i$ and $j$ be nodes in graph $\mathcal{G}$, with nodes $V$, associated with random variables $X_i$ and $X_j$ for $i,j\in V$. Then $i$ is a cause of $j$ if the following three conditions are satisfied:
\begin{itemize}
\item[1.] There is a statistical dependence between $X_i$ and $X_j$,
\item[2.] $X_i$ precedes $X_j$ (in time), 
\item[3.] There is no confounding between $X_i$ and $X_j$ (no alternative explanation for the statistical dependence between $X_i$ and $X_j$). 
\end{itemize}
\end{definition}
\begin{definition}{\rm (Causal intervention, \citet{Woodward:2003,Gillies:2019})}\label{def:causal-intervention}
For $i$ and $j$ in the set of nodes $V$, $i$ is a {\em cause} of $j$ if a change in $X_i$ is followed by a change in the probability distribution of $X_j$. If $i$ is a cause of $j$ we write $i\to j$. (We also sometimes write $X_i\to X_j$.)
\end{definition}

The definitions of causal intervention and causal evidence have been shown to be equivalent, given Assumption \ref{ass:markov-condition} and Assumption \ref{ass:faithfulness-condition} below  \citep{Waldorp:2025}.
%
%\begin{definition}{\rm (Direct cause, \citet{Woodward:2003})}\label{def:direct-cause}
%For $A$ and $B$ in the set $\bf{V}$, $A$ is a {\em direct cause} of $B$ ($A\to B$) if a change in $A$ is followed by a change in the probability distribution of $B$, given that all other variables in $\bf{V}$ are held fixed.
%\end{definition}
%

%%
%\begin{assumption}{\rm (Existence and measurability)}\label{ass:existence-measurability}
%There exists a set $V$ of random variables. We have a set of observed random variables associated with the index set$O\subset V$ and there are random variables that are not measured ($V$ without $O$), and so are not in $O$. The values any random variable can obtain can be distinguished, i.e., a variable can be measured. The values are intervals in the case of continuous variables. 
%\end{assumption}
%%

%
\begin{assumption}{\rm (Causal Markov condition)}\label{ass:markov-condition}
Let $A$, $B$ and $S$ be disjoint sets in $V$. Then, $A$ and $B$ are $d$-separated given $S$ implies that the random variables $X_A$ and $X_B$ are conditionally independent given $X_S$. In symbols we write $A\independent B\mid S \Rightarrow X_A\independent X_B\mid X_S$.
\end{assumption}
\begin{assumption}{\rm (Causal faithfulness condition)}\label{ass:faithfulness-condition}
Let $A$, $B$ and $S$ be disjoint sets in $V$. Then, $X_A$ and $X_B$ are conditionally independent given $X_S$ implies that the nodes in $A$ and $B$ are $d$-separated given $S$. We write $X_A\independent X_B\mid X_S \Rightarrow A\independent B\mid S$
\end{assumption}
\begin{assumption}{\rm (Causal sufficiency)}\label{ass:causal-sufficiency}
A set $O\subseteq V$ is causally sufficient if there is no hidden variable $k\notin O$ such that $k$ is a common cause of any $i$ and $j$, both in $O$, in the graph $\mathcal{G}$ over $V$ nodes. 
\end{assumption}
\begin{assumption}{\rm (No selection bias)}\label{ass:selection-bias}
A set of variables $O\subseteq V$ has no selection bias if there is no hidden variable $k\notin O$ such that $k$ is a common effect of any $i$ and $j$, both in $O$, in the graph $\mathcal{G}$ over $V$ nodes. 
\end{assumption}

Both the causal Markov and faithfulness conditions are standard assumptions of FCI, while causal sufficiency and no selection bias are not assumptions of FCI but of many other algorithms, like the PC algorithm. 

The following is important for obtaining conditional independencies in an undirected graph, which can be obtained from a directed graph, representing the same conditional independence relations \citep[][Lemma 3.21]{Lauritzen96}. Let $\mathcal{G}$ be a directed acyclic graph, as defined above. Then the moral graph $\mathcal{G}^m$ is obtained from $\mathcal{G}$ by keeping all edges, all arrowheads are removed from the edges in $\mathcal{G}$, and any unconnected parents of the same child in $\mathcal{G}$ are connected by an edge in $\mathcal{G}^m$. 

\begin{definition}{(Factorisation undirected graph)}\label{def:undirected-factorisation}
A factorisation of $\mathcal{G}^m$ is a representation of the joint probability density $d\mathbb{P}(X \le x)/d X=p(x)>0$ (strictly positive), in fully connected subsets of the graph (cliques), such that 
\begin{align}\label{eq:undirected-factorisation}
p(x)=\prod_{C\in \mathcal{C}} p_C(x_C),
\end{align}
where $\mathcal{C}$ is the class of cliques of graph $\mathcal{G}^m$. 
\end{definition}

\begin{lemma}{\rm (Hammersley-Clifford \citep{Wainwright:2019,Lauritzen96})}\label{lem:hammersley-clifford}
Let $\mathcal{G}^m$ be the moral graph of the directed acyclic graph $\mathcal{G}$, and let the vector $X$ have strictly positive density. Then the following are equivalent:
\begin{itemize}
\item[(a)] the vector $X$ factorises according to the cliques of $\mathcal{G}^m$, and 
\item[(b)] the vector $X$ is Markov with respect to graph $\mathcal{G}^m$.
\end{itemize}
\end{lemma}

\begin{definition}{(Factorisation directed graph)}\label{def:factorisation-directed}
A directed graphs can be factorisation with respect to the parents (two-node cliques with child) of each node 
\begin{align}\label{eq:directed-factorisation}
p(x) = \prod_{i\in V} p_i(X_i\mid X_{\textsc{pa}(i)}),
\end{align}
where $\textsc{pa}(i)$ is the set of direct causes of node $i$. 
\end{definition}
The factorisation in (\ref{eq:directed-factorisation}) implies the factorisation for undirected graphs in (\ref{eq:undirected-factorisation}) on the moral graph but not vice versa \citep[][Lemma 3.21]{Lauritzen96}. 

A do-intervention can be defined on such a directed factorisation of $p(x^\star)$, with $x^\star=(x_1,\ldots,x_j^\star,\ldots,x_p)$ for $j\in A$ a subset of nodes $A$ of $V$, with particular values $(x_j^\star, j\in A)$. 
\begin{definition}{\citep[do-intervention density,][]{Pearl:2000,Lauritzen:2001}}\label{def:do-intervention-density}
In a do-intervention on nodes in $A\subset V$, the values of the nodes in $A$ are set to $x_j^\star$ and the factors $p(x^\star_j\mid \textsc{pa}(j))$ for $j\in A$ are removed from the factorisation 
\begin{align}
p(x^\star\mid \text{do}(x_A^\star))= \prod_{i\notin A} p_i(x^\star_i\mid X_{\textsc{pa}(i)}).
\end{align}
\end{definition}
This implies that none of the nodes that are parents of nodes in $A$ in the original graph $\mathcal{G}$, have control over the value of nodes in $A$ in the manipulated graph $\mathcal{G}^{\text{do}}$, where the edges into $A$ have been removed \citep{Spirtes:2000}. For example, in the regression context, an intervention on node $i$ such that $\text{do}(X_i=x_i^\star)$, gives 
\begin{align*}
X_i=x_i^\star   
\end{align*}
so that $X_i$ is degenerate and fixed and all influences from other nodes have been removed. 

In contrast, in a soft intervention the graph and its edges remain unchanged but the values of the intervened variables are changed \citep{Eberhardt:2007}. 
\begin{definition}{\citep[Soft intervention density,][]{Eberhardt:2007}}\label{def:soft-intervention} In a soft intervention on nodes $A\subset V$, the values of the nodes in $A$ are adjusted by additional intervention variables $U_i$ and the terms in the factorisation become $p(x_i^\star\mid x_{\sc pa(i)}, u_i)$ with $i\in A$. 
\end{definition}
For example, in the linear case we may obtain as a soft intervention on node $i$ such that $X_i$ is changed by some input $U_i=g(u_i)+r_i$, where $g$ is a continuous bounded function and $u_i\in \mathbb{R}$ and $r_i$ is uniform distributed on $[-1,1]$, for instance. Then the intervention result is
\begin{align*}
    X_i = \beta_0 +\beta_{ij}X_j + U_i +e_i
\end{align*}
This shows that previous influences on $X_i$ remain, but the change in $X_i$ is now also influenced by $U_i$.
%---------------------------------------------------------------------------
\subsection{Difference between regression and causal effects}\label{app:regression-causality}
To illustrate the difference between standard regression coefficients and causal effects, we compare the regression effects and the causal effects for the graph in Figure \ref{fig:causal-graph}(a). The edge weights represent the unique effect of the parent (direct cause). Suppose we are interested in determining the effect of nodes $s$, $u$ and $v$ on $t$. In regression, this effect is obtained by conditioning on all remaining variables. By contrast, in causal modeling we condition only on the parents of the nodes. In Figure \ref{fig:causal-graph}(b), (c) and (d) we see the paths that have a causal effect on $t$. We obtain for the regression and caual effects
\begin{align*}
    &\text{regression}         &\text{causal}  \\
    &\beta_{s\mid uv}= 2        &\beta_{s\mid \varnothing} =0.8 \\
    &\beta_{u\mid sv}=-1        &\beta_{u\mid s} = -1 \\
    &\beta_{v\mid su} = -1      &\beta_{v\mid s} = -1
\end{align*}    
The interesting difference is seen in the effect of $s$ on $t$, where the regression effect is 2 and the causal effect is 0.8. The reason is that in the causal effect, since $s$ does not have any parents, we do not condition on any other node. Hence we add each of the paths in Figure \ref{fig:causal-graph}(b), (c) and (d) to get 0.8. 
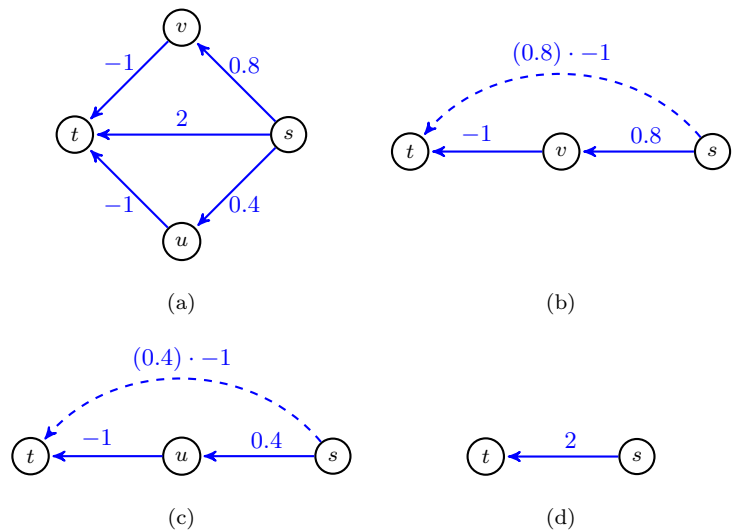
\begin{figure}
\begin{tabular}{c c}
\begin{tikzpicture}[-, >=stealth',shorten >=1pt,auto,node distance=2cm,
  thick,main node/.style={circle,fill=blue!10,draw,font=\footnotesize\sffamily}]
  \tikzstyle{sample} = [circle,fill=blue!20,draw,font=\footnotesize\sffamily]
  \tikzstyle{noSample} = [circle,fill=red!0,draw,font=\footnotesize\sffamily]
  \tikzstyle{sampleEdge} = [font=\sffamily\small,blue]
  \tikzstyle{noSampleEdge} = [node/.style={font=\sffamily\small},red]	
	
  \node[noSample] (v) {$v$};
  \node[noSample] (t) [below left of=v] {$t$};
  \node[noSample] (s) [below right of=v] {$s$};
  \node[noSample] (u) [below right of=t] {$u$};

 \path[sampleEdge,->]
    (s) edge [right] node [above, yshift=0em, xshift=0.4em] {$0.8$} (v)
    (s) edge [right] node [below, yshift=0em, xshift=0.4em] {$0.4$} (u)
    (v) edge [right] node [above, yshift=0em, xshift=-0.4em] {$-1$} (t)
    (u) edge [right] node [below, yshift=0em, xshift=-0.4em] {$-1$} (t)     
    (s) edge [right] node [above] {$2$} (t);

\end{tikzpicture}

&
\raisebox{4em}{
\begin{tikzpicture}[-, >=stealth',shorten >=1pt,auto,node distance=2cm,
  thick,main node/.style={circle,fill=blue!10,draw,font=\footnotesize\sffamily}]
  \tikzstyle{sample} = [circle,fill=blue!20,draw,font=\footnotesize\sffamily]
  \tikzstyle{noSample} = [circle,fill=red!0,draw,font=\footnotesize\sffamily]
  \tikzstyle{sampleEdge} = [font=\sffamily\small,blue]
  \tikzstyle{noSampleEdge} = [node/.style={font=\sffamily\small},red]	
	
  \node[noSample] (s) {$s$};
  \node[noSample] (v) [left of=s] {$v$};
  \node[noSample] (t) [left of=v] {$t$};

 \path[sampleEdge,->]
    (s) edge [right] node [above, yshift=0em, xshift=0.4em] {$0.8$} (v)
    (v) edge [right] node [above, yshift=0em, xshift=-0.4em] {$-1$} (t)
    (s) edge [bend right=50, dashed] node [above] {$(0.8)\cdot -1$} (t);

\end{tikzpicture}
}
\\[1em]
(a) &(b)

\\[1em]

\begin{tikzpicture}[-, >=stealth',shorten >=1pt,auto,node distance=2cm,
  thick,main node/.style={circle,fill=blue!10,draw,font=\footnotesize\sffamily}]
  \tikzstyle{sample} = [circle,fill=blue!20,draw,font=\footnotesize\sffamily]
  \tikzstyle{noSample} = [circle,fill=red!0,draw,font=\footnotesize\sffamily]
  \tikzstyle{sampleEdge} = [font=\sffamily\small,blue]
  \tikzstyle{noSampleEdge} = [node/.style={font=\sffamily\small},red]	
	
  \node[noSample] (s) {$s$};
  \node[noSample] (u) [left of=s] {$u$};
  \node[noSample] (t) [left of=v] {$t$};

 \path[sampleEdge,->]
    (s) edge [right] node [above, yshift=0em, xshift=0.4em] {$0.4$} (v)
    (v) edge [right] node [above, yshift=0em, xshift=-0.4em] {$-1$} (t)
    (s) edge [bend right=50, dashed] node [above] {$(0.4)\cdot -1$} (t);

\end{tikzpicture}

&
\begin{tikzpicture}[-, >=stealth',shorten >=1pt,auto,node distance=2cm,
  thick,main node/.style={circle,fill=blue!10,draw,font=\footnotesize\sffamily}]
  \tikzstyle{sample} = [circle,fill=blue!20,draw,font=\footnotesize\sffamily]
  \tikzstyle{noSample} = [circle,fill=red!0,draw,font=\footnotesize\sffamily]
  \tikzstyle{sampleEdge} = [font=\sffamily\small,blue]
  \tikzstyle{noSampleEdge} = [node/.style={font=\sffamily\small},red]	
	
  \node[noSample] (s) {$s$};
  \node[noSample] (t) [left of=s] {$t$};

 \path[sampleEdge,->]
    (s) edge [right] node [above, yshift=0em, xshift=0.4em] {$2$} (t);

\end{tikzpicture}
\\[1em]
(c) &(d)
\end{tabular}
\caption{Illustration of the difference between regression coefficients and causal effects. The regression coefficient of $s$ on $t$ is $2$ (keeping the other two paths fixed), whereas the causal effect is $0.8$ (see text for details). }
\label{fig:causal-graph}
\end{figure}
%

%---------------------------------------------------------------------------
\section{Markov chains and convergence to equilibrium}\label{app:markov-chains}
We have used material from \citet{Levin:2017} and \citet{Norris:1997} for this section. For further reading we recommend those books and the elegant book by \citet{Haggstrom:2002}. A Markov chain with states, $a$ and $b$ etc, i.e., the possible values for the random variable $X_n$ are $X_n=a$ or $X_n=b$ etc. The random variable is a function of discrete time $n$ and can change state from time $n$ to time $n+1$. This forms the Markov chain $(X_n, n\ge 0)$. The random variable changes state with, so-called, transition (conditional) probability 
\begin{align*}
\mathbb{P}(X_{n+1}=b\mid X_n=a)
\end{align*}
This leads to the transition matrix $P$ with all conditional probabilities, e.g., for a two state model with states $a$ and $b$ we obtain
\begin{align*}
P=
\begin{pmatrix}
\mathbb{P}(X_{n+1}=a\mid X_n=a) 	&\mathbb{P}(X_{n+1}=b\mid X_n=a) \\
\mathbb{P}(X_{n+1}=a\mid X_n=b) 	&\mathbb{P}(X_{n+1}=b\mid X_n=b) 
\end{pmatrix}
\end{align*}
Note that each row sums to 1. This implies that the first row is the conditional distribution of $X_{n+1}$ given $X_n=a$, and the second row is the conditional distribution of $X_{n+1}$ given $X_n=b$. The columns do not necessarily sum to 1. The matrix $P$ is often referred to as a stochastic matrix.

Multiplication of $P$ by itself, i.e., $PP=P^2$, gives the 2-step change.
We can continue the multiplication and writing out $P^n$ for any $n$. We take an initial distribution $\mu_0$ to start the process. For instance, if we start in $a$, then the initial distribution is $\mu_0=(1,0)$, where with probability 1 we start at $a$. Then we take steps by 
\begin{align*}
\mu_1=\mu_0 P\quad\text{and for any $n$}\quad \mu_{n}=\mu_0 P^n
\end{align*}
The question then arises whether at some point the difference between $\mu_n$ and $\mu_{n+1}$ becomes smaller than an arbitrarily small $\varepsilon$. In other words, whether the equation $\mu_n=\mu_0 P^n$ converges such that there is a stationary (equilibrium) distribution $\pi$ such that 
\begin{align*}
\pi=\pi P
\end{align*}
For finite state space Markov chains with non-changing transition probabilities such a stationary distribution $\pi$ often exists \citep[see Example 1.8.1 of][for a nice counterexample]{Norris:1997}. We can obtain the distribution for any $p,q$ in the interval $(0,1)$. 

For general state spaces of some dimension $s$ such easy analysis of stationary distributions is not possible. Let $S$ be a general (finite) state space and let for any states $a,b$ from $S$, $P^n(a,b)$ be the $n$-step transition probability from $a$ to $b$ (i.e., element $a,b$ from the matrix $P^n$). The convergence theorem \citep[][Theorem 4.9]{Robinson:1998} also called the convergence to equilibrium theorem \citep[][Theorem 1.8.3]{Norris:1997} states: if 
\begin{itemize}
\item[(1)] for any states $a$ and $b$ from the state space $S$ there is an $n$ such that $P^n(a,b)>0$ (irreducible), and 
\item[(2)] there exists an $n$ such that $P^n(a,b)>0$ for any states $a$ and $b$ from the state space $S$ (aperiodic)
\end{itemize}
then the difference (in total variation) between the stationary distribution $\pi$ and any row of the transition matrix $P$ tends to 0 at exponential rate. The fact that the convergence to equilibrium is at exponential rate means that often after only a few steps, the rows of the matrix $P^n$ become close to the stationary distribution.

%---------------------------------------------------------------------------
\section{Asymptotic properties of the Markov process}\label{app:asymptotic-properties-linear-Markov}
Here we show that the map
\begin{align*}
    X_{j,t+k} 
    % = f \left( P^kX_t + \sum_{s=0}^{k-1}P^s u \right)
        = f \Bigl( P_k(j,j)X_{j,t} + \sum_{i\ne j} P_k(i,j)X_{i,t} + \sum_{i\in V}C_k(i,j) u_{i,t} \Bigr)
\end{align*}
given in (\ref{eq:tanh-control}) converges. We show that ($i$) the linear map (i.e., $f=\text{id}$ the identity function) is ergodic \citep{Lind:1995,Brin:2002} in Lemma \ref{lem:asymptotically-uncorrelated}, and ($ii$) applying $f$ forms a contraction \citep[see e.g.,][]{Hirsch:2004,Broer:2011,Homburg:2024} of the ergodic map in Lemma \ref{lem:tanh-convergence}, and, hence, is also ergodic \citep{Brin:2002}, and so converges almost surely. The proviso is that we break up the domain on $\mathbb{R}$ into two domains around $0$, so that $[a,b]$ with $a>0$ it is positive or $b<0$ it is  negative. This leads to the bistability and to the possibility of switching from positive to negative or vice versa.

% Lemma dependent sequence convergence ------------------------------------------------------------------------------------------------
\begin{lemma}\label{lem:asymptotically-uncorrelated}
Let $X_{n}=P_nX_{n-1}+U_{n}$ be a linear time varying process with $U_n$ independent non-identically distributed with $\sup_n\mathbb{E}||U_n||<\infty$ and finite variance matrix $\text{var}(U_n)$. Furthermore, $U_n$ is independent of $P_n$, $\mathbb{E} P_n$ converges almost surely to $P$, which is Schur stable. Then $X_n$ is asymptotically uncorrelated and, hence, $\bar{X}_n$ converges almost surely to $\mathbb{E}\bar{X}_n$. 
\end{lemma}
\begin{proof}
Asymptotically uncorrelated is defined as follows \citep[][Definition 3.55]{White01}: $X_n$ is asymptotically uncorrelated if for the covariance of $X_n$ and its shifted version $X_{n+h}$ for $h>0$
\begin{align*}
||\text{cov}(X_{n},X_{n+h})||\le \varrho_h ||\text{var}(X_n)\text{var}(X_{n+h})||^{1/2}
\end{align*}
for $0\le \varrho_h\le 1$ and $\sum_{h=0}^\infty \varrho_h<\infty$ ($\varrho_h$ is an upper bound for the correlation). This implies that the $\varrho_h$ decrease fast enough, and so the correlation decreases to 0 since $\varrho_h$ is an upper bound for the correlation. 

For the variance of $X_n$ we see that for each $n$ and $h\ge 0$
\begin{align*}
\text{var}(X_{n+h}) = \sum_{r=0}^{n+h-1}\text{var}(K_rU_{n+h-r})
\end{align*}
because the $U_n$ are independent. Because the $P_n$ are independent and is also inpendent of the $U_n$, and $\mathbb{E} P_n=P$, we obtain $\mathbb{E}(K_r U_{n-r})=P^r\mu$. Furthermore, changing the index $r \mapsto n+h-s$ for $s\ge 0$ gives
\begin{align*}
\text{var}(K_{n+h-s}U_{s})
%=
%\mathbb{E}_{Y,U}(K_{r}U_{n-r} - \mathbb{E}(K_rU_{n-r}))(K_{r}U_{n-r} - \mathbb{E}(K_rU_{n-r}))^\top
=
P^{n+h-s}\text{var}(U_{s})(P^\top)^{n+h-s}
\end{align*}
Then with $\lambda$ the largest eigenvalue of $P$ for this variance we obtain that 
$$||\text{var}(K_{n+h-s}U_{s}))||\le \beta \lambda^{2(n+h-s)},$$
for some finite $\beta>0$. Because $P$ is Schur stable and $\text{var}(U_n)$ is finite for each $n$, this implies that $\text{var}(X_n)$ converges a.s. and we obtain 
\begin{align*}
||\text{var}(X_n)\text{var}(X_{n+h})||^{1/2} \le n\beta_1\lambda^{2(2n+h)}
%\to \beta\frac{1}{1+\lambda^{2}}
\end{align*}
for some finite $\beta_1>0$. Clearly, since $0\le \lambda <1$, we can take an upper bound for each $h$, $\varrho_h=\beta_2\lambda^{2h}$ with $\beta_2=n\beta_1\lambda^{4n}$, such that $\sum_h \varrho_h=\beta_2/(1-\lambda^2)<\infty$ 
%with an upper bound for the rate $n\beta\lambda^{2n}$. 
%For the covariance of $X_n$ and its shifted version $X_{n+h}$ for $h>0$, recalling that the $U_n$ are independent, we have
%%
%\begin{align*}
%\text{cov}(X_n,X_{n+h}) 
%%=
%%\text{cov}(P_nX_{n-1}+U_{n},P_{n+h}X_{n-1+h}+U_{n+h})
%%=
%%\text{cov}(\sum_{r=0}^{n-1}P_rU_{n-r},\sum_{s=0}^{n-1}P_{s+h}U_{n-s+h})
%=
%\text{cov}(\sum_{r=0}^{n-1}K_rU_{n-r},\sum_{s=0}^{n-1}K_{s+h}U_{n-s+h})
%%&=
%%\text{cov}(U_{n} + K_1U_{n-1}+\cdots,K_{h}U_{n+h}+K_{1+h}U_{n-1+h}+\cdots)\\
%&=
%\sum_{r=h}^{n-1}\text{var}(K_{h:r}U_{n-r})
%\end{align*}
%%
%where $K_{h:r}=P_hP_{h+1}\cdots P_{r}$. 
%Because the covariance can be expressed in terms of the variance, we easily see that 
%%Let the correlation function be defined as  $R_n(h) = D^{-1/2}\text{cov}(X_n,X_{n+h})D^{-1/2}$, where $D$ is the diagonal matrix $\text{diag}(\text{var}(X_n))$. Then 
%%%
%%\begin{align*}
%%\text{cor}(X_{n,i},X_{n+h,j})
%%%=\frac{\text{cov}(X_{n,i},X_{n+h,j})}{ (\text{var}(X_{n,i})\text{var}(X_{n+h,j}))^{1/2}}
%%=\frac{\sum_{r=h}^{n-1}\text{var}(K_{h:r}U_{n-r})_{ij}}{ (\sum_{r=0}^{n-1}\text{var}(K_rU_{n-r})_{ii}\sum_{r=0}^{n-1}\text{var}(K_rU_{n-r})_{jj})^{1/2}}
%%\end{align*}
%%%
%%$(R_n(0))_{ij} \le (R_n(h))_{ij}/h$ for $h>0$ and $i\ne j$. 
%Hence, the correlations decrease at least at rate $1/h$. Therefore, we can take any upper bound $\varrho_{h}=1/h^{1+\delta}$, for some $\delta>0$, such that $\sum_{h}\varrho_h<\infty$. 

We assumed that $\sup_n\mathbb{E}||U_n||$ is finite, and so we have that
\begin{align*}
\mathbb{E}||X_n||
%=\mathbb{E}P_n\mathbb{E}X_{n-1}+\mathbb{E}U_n=P\mathbb{E}X_{n-1} + \mu_n
\le \mathbb{E}||K_n x_0|| + \mathbb{E}||\sum_{r=0}^{n-1}K_r U_{n-r}||
%\le \sup_n \mathbb{E}||U_n|| \sum_{r=0}^{n-1}\mathbb{E}||K_r||
\end{align*}
where $K_r=P_0 P_1\cdots P_r$. Because each $P_n$ is Schur stable, $||K_r||$ is a decreasing sequence in $r$. Therefore,  for $n\ge n_0$
\begin{align*}
%\mathbb{E}||X_n||
%=\mathbb{E}P_n\mathbb{E}X_{n-1}+\mathbb{E}U_n=P\mathbb{E}X_{n-1} + \mu_n
\mathbb{E}||K_n x_0|| + \mathbb{E}||\sum_{r=0}^{n-1}K_r U_{n-r}||
\le \sup_n \mathbb{E}||U_n|| \sum_{r=0}^{n-1}\mathbb{E}||K_r||
\end{align*}
and so, the expectation $\mathbb{E}||X_n||$ is finite and $\mathbb{E}X_n$ converges almost surely. We obtain for all $n\ge n_0$
\begin{align*}
\mathbb{E} X_n
=\mathbb{E}P_n\mathbb{E}X_{n-1}+\mathbb{E}U_n=P^n x_0 + \sum_{r=0}^{n-1}P^r\mu_{n-r}
\end{align*}
%
%Then 
%%
%\begin{align*}
%X_j - \mathbb{E} X_j
%=(K_j-P^j )x_0 + \sum_{r=0}^{n-1}(K_rU_{n-r}-Q^r\mu_{n-r})
%\end{align*}
%%
%The term $K_j-Q^j$ converges to 0, since both $K_j$ and $Q^j$ do. 
Let $Y_j=\sum_{r=0}^{j-1}K_rU_{n-r}$ and $\nu_j= \sum_{r=0}^{j-1}P^r\mu_{n-r}$. Then $\mathbb{E}(Y_j-\nu_j)=0$ for each $j$ and 
\begin{align*}
\sum_{j=1}^n \mathbb{E}||X_j-\nu_j||^2 
\le \sum_{j=1}^{n}\text{tr}( \text{var}(X_j)) 
\le n\beta_1\lambda^{4n}
\end{align*}
Hence, we obtain for all $n\ge n_0$ that 
\begin{align*}
\bar{X}_n=\frac{1}{n}\sum_{j=1}^n X_j  = \frac{1}{n}\sum_{j=1}^n \nu_j=\bar{\nu}
\end{align*}
almost surely.
\end{proof}
%
%-----------------------------------------------------------------------------------------------------------------------------------------------------------

%
\begin{lemma}\label{lem:tanh-convergence}
The map $f:[a,b]\to[-1,1]$ with $x\mapsto\tanh(x)$ and either $[a,b]$ with $a>0$ or $b<0$, with any (possibly stochastic) linear function for $x$ and the assumptions in Lemma \ref{lem:asymptotically-uncorrelated}, converges uniformly almost surely to either one of two fixed points of $f$ on $[a,b]$ with $a>0$ or $b<0$. 
\end{lemma}
\begin{proof}
We take $a>0$ so that we consider the positive domain; for $b<0$ the proof is nearly identical. 

The map $f(x)=\tanh(x)$ is $1$-Lipschitz continuous, since for all $x,y\in [a,b]$, $|\tanh'(x)|=|1-\tanh^2(x)|\le 1$. Hence, we obtain contraction
\begin{align*}
| \tanh(x) - \tanh(y) | \le |x-y|.
\end{align*}
Because the mapping is recursive, i.e., $f^n(x_0)=f(f(\cdots f(x_0)))$, $n$ times, we can apply the contraction theorem \citep[e.g.][]{Broer:2011,Homburg:2024}; and by the assumptions in Lemma \ref{lem:asymptotically-uncorrelated}, all values are bounded. Hence, we obtain convergence almost surely. 
\end{proof}
\begin{remark}
A matrix $P$ is Schur stable if the absolute value of the maximal eigenvalue is $<1$. One way to assure that $P$ (with positive entries) is Schur stable is to require that \citep[][Exercise 8.3.7(b)]{Meyer:2000}
\begin{itemize}
\item[($1$)] all rows of $P$ sum to at most 1 and at least one row sums to $<1$ (substochastic), and 
\item[($2$)] the underlying network of $P$ is such that each node in the network can reach the node which has sum $<1$ (irreducible). 
\end{itemize}
These conditions together result in a transition matrix $P$ that has eigenvalues that are in the interval $(-1,1)$, and so are Schur stable.
\end{remark}
\begin{remark}\label{rem:control-theory-target}
We obtain the result for $k$ finite time steps that $X_n=P^n x_0 + (I-P^k)(I-P)^{-1}u$, where $P$ is the transition matrix, (used in (\ref{eq:tanh-control})) from the result \citep[][Lemma 2.1]{Galor:2007}
\begin{align*}
\sum_{n=0}^{k-1}P^n=(I-P^k)(I-P)^{-1}
\end{align*}
This is how we obtain the elements $C_k(i,j)=[(I-P^k)(I-P)^{-1}]_{ij}$ in (\ref{eq:tanh-control}) and above. So, the target $u$ set in control theory is achieved in $k$ steps by the given equation. 
\end{remark}
%
%---------------------------------------------------------------------------
\section{Distance measures in linear Markov processes}\label{app:distance-measures-linear-markov}
We briefly discuss two types of distances defined on the linear processes (as in \ref{eq:markov-dynamics})) that are related to what we use in the non-linear setting. 

The first is to determine which node is most effecient to start with in an intervention. We therefore consider the distance to get from node $i$ to node $j$ in time $t$, and then which node has the strongest influence. The distance between nodes $i$ and $j$ considering all paths between them is \citep{Coifman:2006}
\begin{align}
    D_n^2(i,j)=\max_{t\in T}||P_t(i,\cdot) - P_t(j,\cdot)||^2=
        \max_{t\in T}\sum_{k\in V}(P_t(i,k) - P_t(j,k))^2
\end{align}
The distance measure can be interpreted as the largest impact of node $i$ on $j$ over any of the time points in range $T$, considering the different causal paths between them. We do this for all nodes and then the largest impact it has on all other nodes could be seen as a good candidate to start with the interventions. The distance across all nodes is then defined as 
\begin{align}\label{eq:node-distance}
    D_n(V) = \max_{i\in V}\left\{\max_{t\in T}\sum_{k\in V}(P_t(i,k) - P_t(j,k))^2\right\}
\end{align}
where we maximise over all nodes to determine the largest influence. 

Our second distance is about the difference between an unchanged kernel and a kernel with intervention (changing a connection or influence). 
To determine whether the intervention on a connection between nodes will make a difference and at what times, we use a Markov distance between the non-intervened and the intervened kernels, referred to as $K$ and $\tilde{K}$, respectively. We draw inspiration a metric that uses a continuous-time Markov chain \citep{Hammond:2013}, which we apply to a discrete-time Markov chain. The basic idea is that for each time $t$ we determine the difference between $P^t$ and $\tilde{P}^t$ and determine the Frobenius norm between them; then we determine the maximal difference in this norm. That is, the graph diffusion distance is defined as
\begin{align}\label{eq:distance}
    D_c^2(P,\tilde{P}) =\max_{t\in T}||P^t-\tilde{P}^t||^2=\max_{t\in T}\sum_{i,j\in V}(P_t(i,j)-\tilde{P}_{t}(i,j))^2
\end{align}
where $T$ is a finite subset of $\mathbb{N}$, and $P_t(i,j)$ is the $ij$th element of $P^t$. The distance $D_c$ can be interpreted as the maximum difference between the paths using either the original $P$ or the intervened version $\tilde{P}$.

%---------------------------------------------------------------------------
\section{Details of analyses in \texttt{R}}\label{app:analyses-R}
We provide the details for obtaining the causal graph, the causal effects, the Markov process, and the intervention metrics. 

\subsection{Causal graph and causal effects}
The algorithm FCI is part of the package \texttt{pcalg}. 
We used the assumption of Gaussian variables for the covariance matrix in (\ref{eq:covariance-time-series}). The test level was set at $0.05$ and we set $B=0$. One motivation for this is that there are several hours in between each measurement, so that each observation can be considered a time-integrated recollection of the period. An example of a causal graph estimated with one lag (i.e., an AR(1) model) is shown in Figure \ref{fig:network-p2-lag1}.
\begin{figure}[t]
\begin{tabular}{c @{\hspace{5em}} p{0.4\textwidth}}
\raisebox{-10em}{
\begin{tikzpicture}[-, >=stealth',shorten >=1pt,auto,node distance=2cm,
  thick,main node/.style={circle,fill=blue!10,draw,font=\footnotesize\sffamily}]
  \tikzstyle{sample} = [circle,fill=blue!20,draw,font=\footnotesize\sffamily]
  \tikzstyle{noSample} = [circle,fill=red!0,draw,font=\footnotesize\sffamily]
  \tikzstyle{sampleEdge} = [font=\sffamily\small,blue]
  \tikzstyle{noSampleEdge} = [node/.style={font=\sffamily\small},red]	
	
  \node[noSample] (anx) [] {$anx$};
  \node[noSample] (gui) [below left of=anx] {$gui$};
  \node[noSample] (str) [below right of=anx] {$str$};
  \node[noSample] (con) [right of=str] {$con$};
  \node[noSample] (sad) [below right of=gui] {$sad$};
  \node[noSample] (irr) [right of=sad] {$irr$};
  \node[noSample] (avo) [right of=irr] {$avo$};  
  \node[noSample] (wor) [left of=sad] {$wor$};

 \path[sampleEdge,<->]
    (anx) edge [right] node [xshift=0.2em] {} (str)  %above, yshift=0em, xshift=0.4em
    (str) edge [right] node [above] {} (con)
    (con) edge [right] node [] {} (avo)
    (irr) edge [right] node [below] {} (avo)
    (irr) edge [right] node [below] {} (avo)
    (irr) edge [right] node [below] {} (sad)
    (sad) edge [right] node [above] {} (con)
    (str) edge [right] node [left, yshift=0.5em, xshift=0.3em] {} (gui)
    (irr) edge [right] node [xshift=0.2em] {} (gui)
    (wor) edge [bend right] node [below] {} (avo);

\path[sampleEdge,->]
    (irr) edge [right] node [] {} (str)
    (gui) edge [right] node [left] {} (anx);

\end{tikzpicture}
}
&
\vspace{-3em}
\begin{itemize}
\item[$avo =$] avoidance 
\item[$wor =$] worry 
\item[$anx =$] anxious 
\item[$str =$] stressed 
\item[$con =$] concentration
\item[$gui =$] guilty
\item[$sad =$] sad
\item[$irr =$] irritability
\end{itemize}

\end{tabular}
\caption{Graph of the eight variables; result of applying FCI to the time series data where $B$ shifts one lag. A directed edge means there is evidence for a causal relation, and a bidirected edge means there is possibly a confound. The numbers on the edges refer to the estimated causal effects (see main text for explanation).}
\label{fig:network-p2-lag1}
\end{figure}
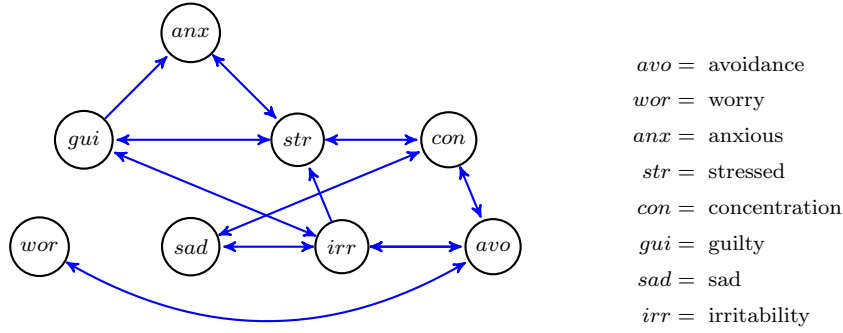

The causal effects were estiamted with the \texttt{ida} \citep[intervention calculus when DAG is absent][]{Maathuis:2010}. The causal graph obtained with FCI was given as input to \texttt{ida} with the value of 0.2. We estimated the multiset (i.e., a set with possible multiplicity of the same values) of effects in (\ref{eq:average-causal-effect}), which uses the assumption of Gaussian variables. For each node separately we estimated the multiset on all of the other nodes. Then we chose the smallest value in the multiset. 

%----------------------
\subsection{Markov process}
The linear Markov process using the transition matrix $P_\mathcal{G}$ was obtained from the kernel $K_\mathcal{G}$ in (\ref{eq:kernel-transition}) described in Section \ref{sec:nonlinear-dynamic-markov-process}. The coefficients were then given their original sign after obtaining $P_\mathcal{G}$.

To generate the time series we used the function in (\ref{eq:tanh-control})
\begin{align}
    X_{j,t+k} 
    % = f \left( P^kX_t + \sum_{s=0}^{k-1}P^s u \right)
        = f \Bigl( \alpha P_k(j,j)X_{j,t} + \gamma\sum_{i\ne j} P_k(i,j)X_{i,t} + \sum_{i\in V}C_k(i,j) u_{i,t} \Bigr)
\end{align}
where $P_k(i,j)=(P^k_\mathcal{G})_{ij}$ and $C_k(i,j)=[(I-P^k_\mathcal{G})(I-P_\mathcal{G})^{-1}]_{ij}$, and $\alpha=\gamma=2$. This setting ensured is similar to setting the inverse temperature in the Ising mean field approximation to 2, leading to bistability. An example of this process with a transient intervention on two nodes is shown in Figure \ref{fig:intervention-35}.
\begin{figure}\centering
    \pgfimage[width=\textwidth]{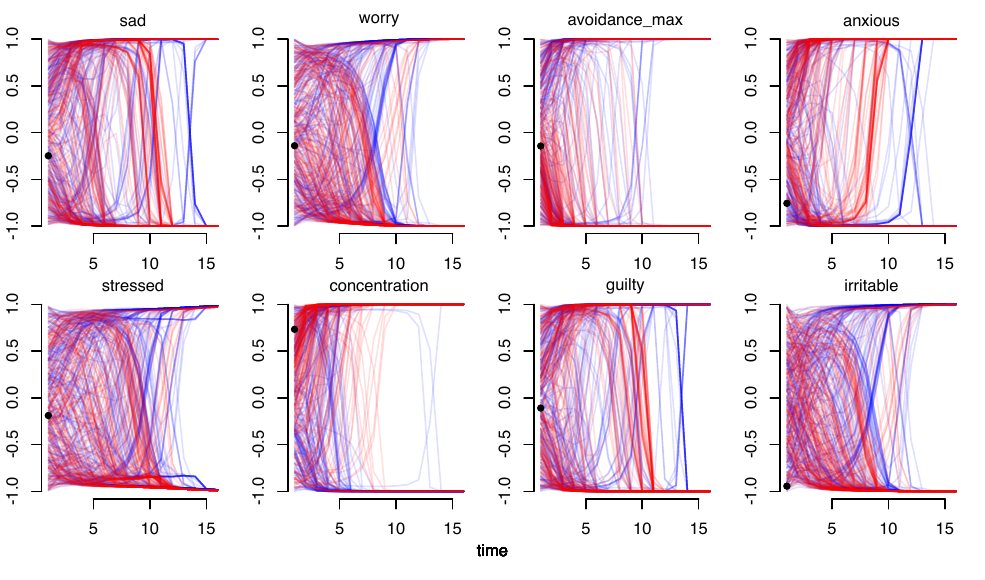}
\caption{ A transient intervention on variables \textit{avoidance} and \textit{concentration} simultaneously. Settings are otherwise the same as in Figure \ref{fig:intervention-stress}. }
\label{fig:intervention-35}
\end{figure}

The Gibbs free energy is obtained from the mean field approximation of the Ising model \citep{Plischke:1994}. Assuming that the correlations between the nodes is relatively weak, we obtain 
\begin{align*}
    G(x) = c_1 x^2 - c_2\log(2\cosh(2c_3x+\gamma))
\end{align*}
where $c_1$, $c_2$ and $c_3$ are constants, and $\gamma$ represents the coupling. We used $c_1=c_2=c_3=1$. The landscape plots in Figure \ref{fig:free-energy-over-time} were then created with input for $x$ the interval $[-3,3]$ and for $\gamma=X_{t+k}$ for different $k$. 

\subsection{Intervention metrics}
The first metric is the average of each of the control and intervention condition separately, at the last time point. The difference between these averages determines the effect size. The second metric is determined by the absolute difference between the averages of control and intervention conditions, for each time and variable separately. Then, the effect is summed across all time points and all variables, and this gives the sum effect.

%----------------------------------
\section{Graphs of five patients}
\label{app:results-5patients}

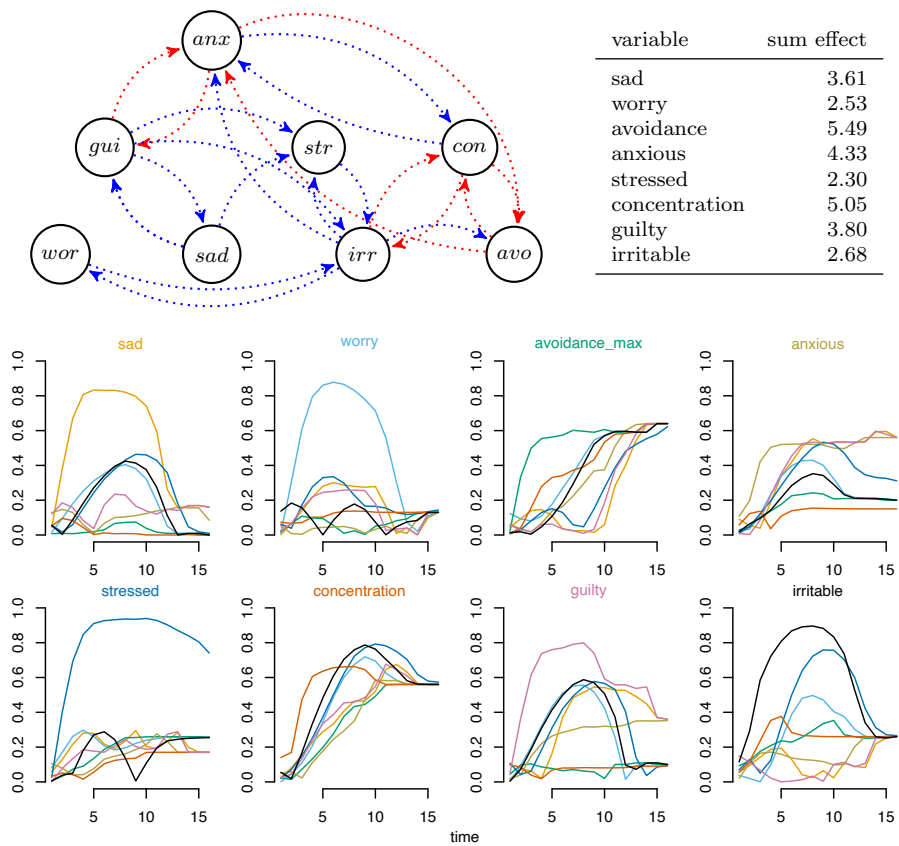
\begin{figure}[t]
\begin{tabular}{c @{\hspace{2em}} p{0.4\textwidth}}
\raisebox{-10em}{
\begin{tikzpicture}[-, >=stealth',shorten >=1pt,auto,node distance=2cm,
  thick,main node/.style={circle,fill=blue!10,draw,font=\footnotesize\sffamily}]
  \tikzstyle{sample} = [circle,fill=blue!20,draw,font=\footnotesize\sffamily]
  \tikzstyle{noSample} = [circle,fill=red!0,draw,font=\footnotesize\sffamily]
  \tikzstyle{sampleEdge} = [font=\sffamily\small,blue]
  \tikzstyle{noSampleEdge} = [node/.style={font=\sffamily\small},red]	
	
  \node[noSample] (anx) [] {$anx$};
  \node[noSample] (gui) [below left of=anx] {$gui$};
  \node[noSample] (str) [below right of=anx] {$str$};
  \node[noSample] (con) [right of=str] {$con$};
  \node[noSample] (sad) [below right of=gui] {$sad$};
  \node[noSample] (irr) [right of=sad] {$irr$};
  \node[noSample] (avo) [right of=irr] {$avo$};  
  \node[noSample] (wor) [left of=sad] {$wor$};

\path[sampleEdge,->,dotted,red]
 (avo) edge [bend left=30] node [xshift=0.2em] {} (anx)
 (avo) edge [bend left=30] node [xshift=0.2em] {} (con)
 (anx) edge [bend left=60] node [xshift=0.2em] {} (avo)
 (anx) edge [bend left=40] node [xshift=0.2em] {} (gui)
 (con) edge [bend left=30] node [xshift=0.2em] {} (avo)
 (con) edge [bend left=30] node [xshift=0.2em] {} (irr)
 (gui) edge [bend left=30] node [xshift=0.2em] {} (anx)
 (irr) edge [bend left=30] node [xshift=0.2em] {} (con);
 %above, yshift=0em, xshift=0.4em

 \path[sampleEdge,->,dotted]
 (sad) edge [bend left] node [xshift=0.2em] {} (gui)
 (anx) edge [bend left] node [xshift=0.2em] {} (con)
 (sad) edge [bend left] node [xshift=0.2em] {} (str)
 (sad) edge [bend left] node [xshift=0.2em] {} (gui)
 (wor) edge [bend right=20] node [xshift=0.2em] {} (irr)
 (str) edge [bend left=30] node [xshift=0.2em] {} (irr)
 (con) edge [bend left=15] node [xshift=0.2em] {} (anx)
 (gui) edge [bend left=30] node [xshift=0.2em] {} (sad)
 (gui) edge [bend left=30] node [xshift=0.2em] {} (str)
 (gui) edge [bend left=30] node [xshift=0.2em] {} (irr)
 (irr) edge [bend left=30] node [xshift=0.2em] {} (wor)
 (irr) edge [bend left=30] node [xshift=0.2em] {} (avo)
 (irr) edge [bend left=30] node [xshift=0.2em] {} (anx)
 (irr) edge [bend left=30] node [xshift=0.2em] {} (str)
 ;
 %above, yshift=0em, xshift=0.4em
    
\end{tikzpicture}
}
&
\vspace{-3em}
\begin{tabular}{l @{\hspace{1em}} r}
variable    &sum effect\\ \midrule
sad         &3.61\\ 
worry       &2.53\\
avoidance   &5.49\\
anxious     &4.33\\
stressed    &2.30\\
concentration &5.05\\
guilty      &3.80\\
irritable   &2.68\\
\midrule
\end{tabular}

\end{tabular}
\\
\centering
    \pgfimage[width=\textwidth]{interv-effect-time-p1}
\caption{patient 1. These results are the same as presented in the main text in Sections \ref{sec:causal-graph}, \ref{sec:nonlinear-dynamic-markov-process} and \ref{sec:interventions-comparisons}. }
\label{fig:intervention-effect-time-p2-app}
\end{figure}
\begin{figure}[t]
\begin{tabular}{c @{\hspace{2em}} p{0.4\textwidth}}
\raisebox{-10em}{
\begin{tikzpicture}[-, >=stealth',shorten >=1pt,auto,node distance=2cm,
  thick,main node/.style={circle,fill=blue!10,draw,font=\footnotesize\sffamily}]
  \tikzstyle{sample} = [circle,fill=blue!20,draw,font=\footnotesize\sffamily]
  \tikzstyle{noSample} = [circle,fill=red!0,draw,font=\footnotesize\sffamily]
  \tikzstyle{sampleEdge} = [font=\sffamily\small,blue]
  \tikzstyle{noSampleEdge} = [node/.style={font=\sffamily\small},red]	
	
  \node[noSample] (anx) [] {$anx$};
  \node[noSample] (gui) [below left of=anx] {$gui$};
  \node[noSample] (str) [below right of=anx] {$str$};
  \node[noSample] (con) [right of=str] {$con$};
  \node[noSample] (sad) [below right of=gui] {$sad$};
  \node[noSample] (irr) [right of=sad] {$irr$};
  \node[noSample] (avo) [right of=irr] {$avo$};  
  \node[noSample] (wor) [left of=sad] {$wor$};

\path[sampleEdge,->,dotted,red]
 (avo) edge [bend left=30] node [xshift=0.2em] {} (anx)
 (avo) edge [bend left=30] node [xshift=0.2em] {} (con)
 (anx) edge [bend left=60] node [xshift=0.2em] {} (avo)
 (anx) edge [bend left=15] node [xshift=0.2em] {} (avo)
;
 %above, yshift=0em, xshift=0.4em

 \path[sampleEdge,->,dotted]
 (sad) edge [bend left] node [xshift=0.2em] {} (anx)
 (sad) edge [bend right=10] node [xshift=0.2em] {} (con)
 (sad) edge [bend left] node [xshift=0.2em] {} (str)
 (wor) edge [bend right=20] node [xshift=0.2em] {} (avo)
 (wor) edge [bend right=20] node [xshift=0.2em] {} (irr)
 (avo) edge [bend right=20] node [xshift=0.2em] {} (wor)
 (avo) edge [bend right=20] node [xshift=0.2em] {} (str)
 (anx) edge [bend left=30] node [xshift=0.2em] {} (sad)(anx) edge [bend left=30] node [xshift=0.2em] {} (gui)
 (str) edge [bend left=30] node [xshift=0.2em] {} (sad)
 (str) edge [bend left=30] node [xshift=0.2em] {} (avo)
 (str) edge [bend left=30] node [xshift=0.2em] {} (anx)
 (str) edge [bend left=30] node [xshift=0.2em] {} (gui)
 (con) edge [bend right=10] node [xshift=0.2em] {} (sad)
 (con) edge [bend right=30] node [xshift=0.2em] {} (anx)
 (con) edge [bend left=30] node [xshift=0.2em] {} (str)
 (con) edge [bend left=30] node [xshift=0.2em] {} (gui)
 (con) edge [bend left=30] node [xshift=0.2em] {} (irr)
 (irr) edge [bend left=30] node [xshift=0.2em] {} (wor)
 (irr) edge [bend left=30] node [xshift=0.2em] {} (con)
 (irr) edge [bend left=30] node [xshift=0.2em] {} (str)
 ;
 %above, yshift=0em, xshift=0.4em
    
\end{tikzpicture}
}
&
\vspace{-3em}
\begin{tabular}{l @{\hspace{1em}} r}
variable    &sum effect\\ \midrule
sad         &1.83\\ 
worry       &2.44\\
avoidance   &3.59\\
anxious     &2.83\\
stressed    &3.03\\
concentration &0.92\\
guilty      &1.54\\
irritable   &1.35\\
\midrule
\end{tabular}

\end{tabular}
\\
\centering
    \pgfimage[width=\textwidth]{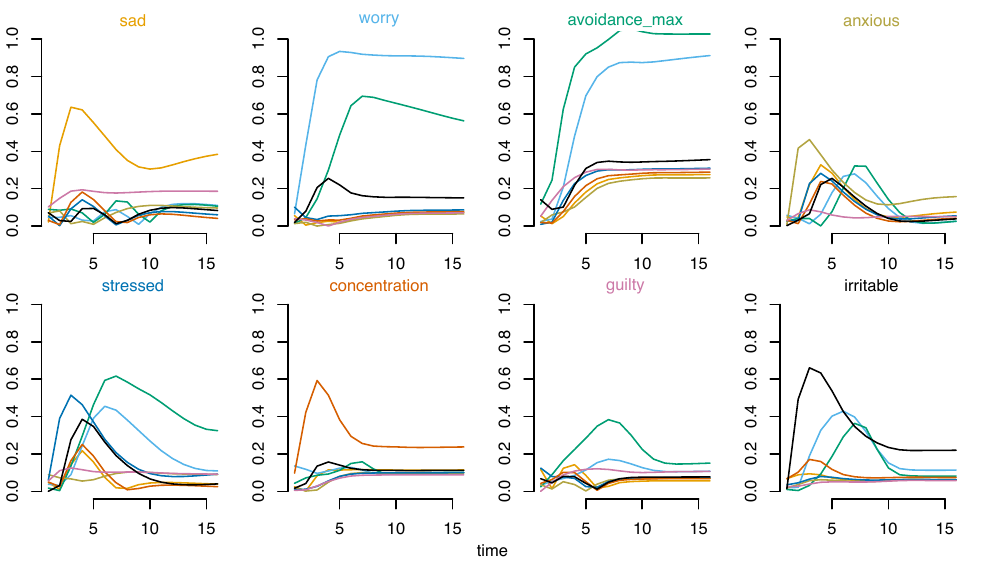}
\caption{ patient 2 }
\label{fig:intervention-effect-time-p2-app}
\end{figure}
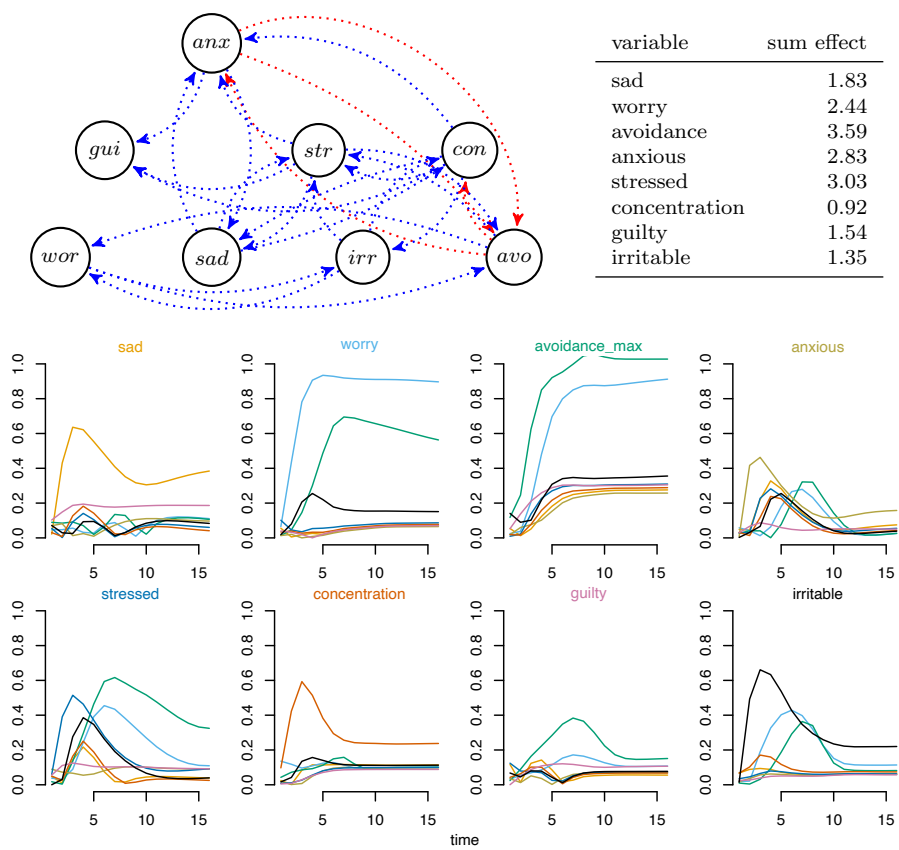
\begin{figure}[t]
\begin{tabular}{c @{\hspace{2em}} p{0.4\textwidth}}
\raisebox{-10em}{
\begin{tikzpicture}[-, >=stealth',shorten >=1pt,auto,node distance=2cm,
  thick,main node/.style={circle,fill=blue!10,draw,font=\footnotesize\sffamily}]
  \tikzstyle{sample} = [circle,fill=blue!20,draw,font=\footnotesize\sffamily]
  \tikzstyle{noSample} = [circle,fill=red!0,draw,font=\footnotesize\sffamily]
  \tikzstyle{sampleEdge} = [font=\sffamily\small,blue]
  \tikzstyle{noSampleEdge} = [node/.style={font=\sffamily\small},red]	
	
  \node[noSample] (anx) [] {$anx$};
  \node[noSample] (gui) [below left of=anx] {$gui$};
  \node[noSample] (str) [below right of=anx] {$str$};
  \node[noSample] (con) [right of=str] {$con$};
  \node[noSample] (sad) [below right of=gui] {$sad$};
  \node[noSample] (irr) [right of=sad] {$irr$};
  \node[noSample] (avo) [right of=irr] {$avo$};  
  \node[noSample] (wor) [left of=sad] {$wor$};

\path[sampleEdge,->,dotted,red]
 (sad) edge [bend left=10] node [xshift=0.2em] {} (con)
 (avo) edge [bend left=30] node [xshift=0.2em] {} (con)
 (anx) edge [bend left=30] node [xshift=0.2em] {} (con)
 (str) edge [bend left=30] node [xshift=0.2em] {} (con)
 (con) edge [bend right=40] node [xshift=0.2em] {} (wor)
 (con) edge [bend left=30] node [xshift=0.2em] {} (avo)
 (con) edge [bend left=30] node [xshift=0.2em] {} (str)
 (con) edge [bend left=30] node [xshift=0.2em] {} (gui)
 (gui) edge [bend left=30] node [xshift=0.2em] {} (con)
 (irr) edge [bend left=30] node [xshift=0.2em] {} (con);
 %above, yshift=0em, xshift=0.4em

 \path[sampleEdge,->,dotted]
 (sad) edge [bend left] node [xshift=0.2em] {} (gui)
 (sad) edge [bend left] node [xshift=0.2em] {} (wor)
 (sad) edge [bend right] node [xshift=0.2em] {} (avo)
 (sad) edge [bend left] node [xshift=0.2em] {} (anx)
 (sad) edge [bend left] node [xshift=0.2em] {} (str)
 (sad) edge [bend left] node [xshift=0.2em] {} (gui)
 (wor) edge [bend right=20] node [xshift=0.2em] {} (gui)
 (wor) edge [bend right=20] node [xshift=0.2em] {} (irr)
 (wor) edge [bend right=20] node [xshift=0.2em] {} (irr)
 (avo) edge [bend left=30] node [xshift=0.2em] {} (sad)
 (avo) edge [bend left=30] node [xshift=0.2em] {} (wor)
 (avo) edge [bend left=30] node [xshift=0.2em] {} (str)
 (avo) edge [bend left=30] node [xshift=0.2em] {} (irr)
 (anx) edge [bend left=30] node [xshift=0.2em] {} (sad)
 (anx) edge [bend left=30] node [xshift=0.2em] {} (wor)
 (anx) edge [bend left=40] node [xshift=0.5em] {} (irr)
 (str) edge [bend left=30] node [xshift=0.2em] {} (sad)
 (str) edge [bend left=30] node [xshift=0.2em] {} (avo)
 (con) edge [bend left=15] node [xshift=0.2em] {} (anx)
 (gui) edge [bend left=30] node [xshift=0.2em] {} (sad)
 (gui) edge [bend left=30] node [xshift=0.2em] {} (wor)
 (gui) edge [bend left=30] node [xshift=0.2em] {} (str)
 (gui) edge [bend left=30] node [xshift=0.2em] {} (irr)
 (irr) edge [bend left=30] node [xshift=0.2em] {} (wor)
 (irr) edge [bend left=10] node [xshift=0.2em] {} (gui)
 ;
 %above, yshift=0em, xshift=0.4em
    
\end{tikzpicture}
}
&
\vspace{-3em}
\begin{tabular}{l @{\hspace{1em}} r}
variable    &sum effect\\ \midrule
sad         &3.00\\ 
worry       &2.63\\
avoidance   &3.40\\
anxious     &2.02\\
stressed    &4.67\\
concentration &6.66\\
guilty      &3.51\\
irritable   &2.61\\
\midrule
\end{tabular}
\end{tabular}
\\
\centering
    \pgfimage[width=\textwidth]{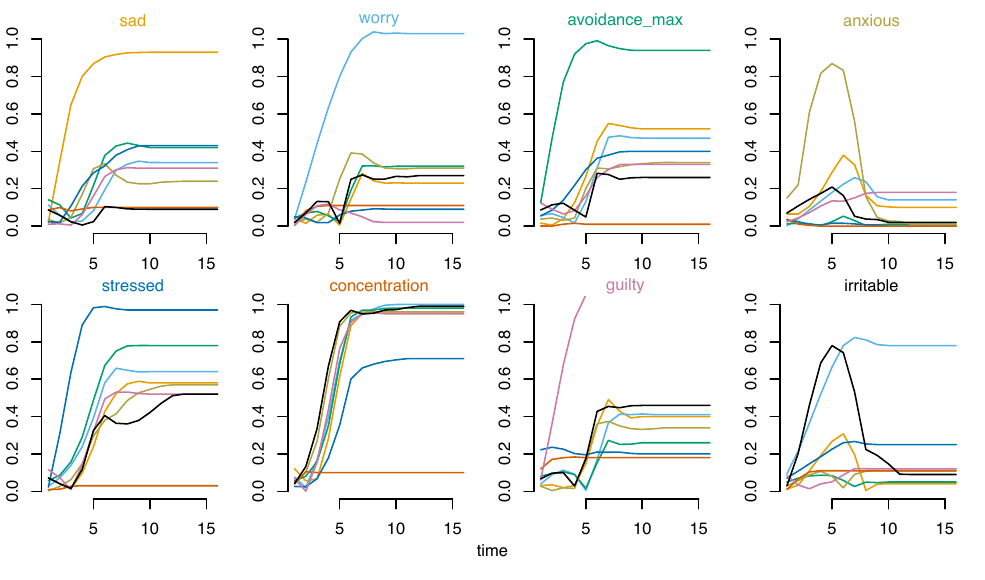}
\caption{ patient 3. Where we used the settings $\alpha=2$, $\gamma=09.1$, and for the causal graph $\alpha=0.1$}
\label{fig:intervention-effect-time-p3}
\end{figure}
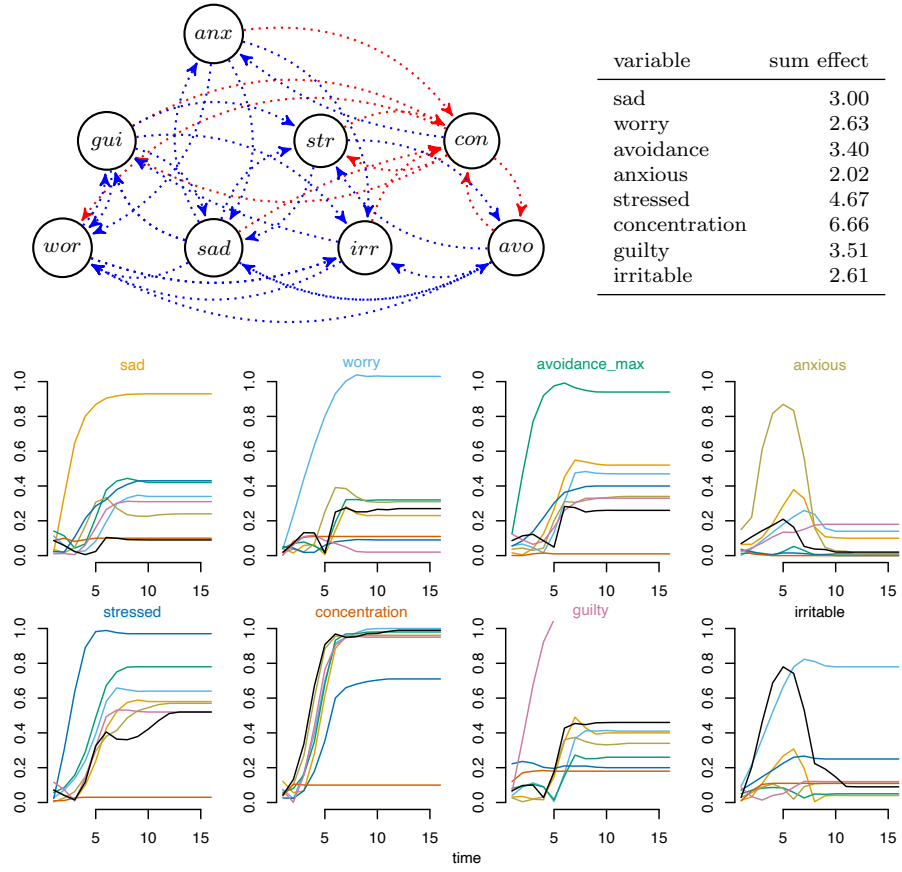
\begin{figure}[t]
\begin{tabular}{c @{\hspace{2em}} p{0.4\textwidth}}
\raisebox{-10em}{
\begin{tikzpicture}[-, >=stealth',shorten >=1pt,auto,node distance=2cm,
  thick,main node/.style={circle,fill=blue!10,draw,font=\footnotesize\sffamily}]
  \tikzstyle{sample} = [circle,fill=blue!20,draw,font=\footnotesize\sffamily]
  \tikzstyle{noSample} = [circle,fill=red!0,draw,font=\footnotesize\sffamily]
  \tikzstyle{sampleEdge} = [font=\sffamily\small,blue]
  \tikzstyle{noSampleEdge} = [node/.style={font=\sffamily\small},red]	
	
  \node[noSample] (anx) [] {$anx$};
  \node[noSample] (gui) [below left of=anx] {$gui$};
  \node[noSample] (str) [below right of=anx] {$str$};
  \node[noSample] (con) [right of=str] {$con$};
  \node[noSample] (sad) [below right of=gui] {$sad$};
  \node[noSample] (irr) [right of=sad] {$irr$};
  \node[noSample] (avo) [right of=irr] {$avo$};  
  \node[noSample] (wor) [left of=sad] {$wor$};

\path[sampleEdge,->,dotted,red]
 (con) edge [bend left=15] node [xshift=0.2em] {} (anx)
 (con) edge [bend left=30] node [xshift=0.2em] {} (gui)
 (con) edge [bend left=15] node [xshift=0.2em] {} (irr)
 (irr) edge [bend left=10] node [xshift=0.2em] {} (con)
;
 %above, yshift=0em, xshift=0.4em

 \path[sampleEdge,->,dotted]
 (sad) edge [bend left] node [xshift=0.2em] {} (wor)
 (sad) edge [bend left=10] node [xshift=0.2em] {} (con)
 (wor) edge [bend left=20] node [xshift=0.2em] {} (sad)
 (avo) edge [bend left=30] node [xshift=0.2em] {} (con)
 (str) edge [bend left=30] node [xshift=0.2em] {} (anx)
 (str) edge [bend left=30] node [xshift=0.2em] {} (irr)
 (irr) edge [bend left=30] node [xshift=0.2em] {} (str)
 ;
 %above, yshift=0em, xshift=0.4em
    
\end{tikzpicture}
}
&
\vspace{-3em}
\begin{tabular}{l @{\hspace{1em}} r}
variable    &sum effect\\ \midrule
sad         &2.43\\ 
worry       &2.72\\
avoidance   &3.20\\
anxious     &2.33\\
stressed    &2.10\\
concentration &5.83\\
guilty      &2.37\\
irritable   &3.26\\
\midrule
\end{tabular}
\end{tabular}
\\
\centering
    \pgfimage[width=\textwidth]{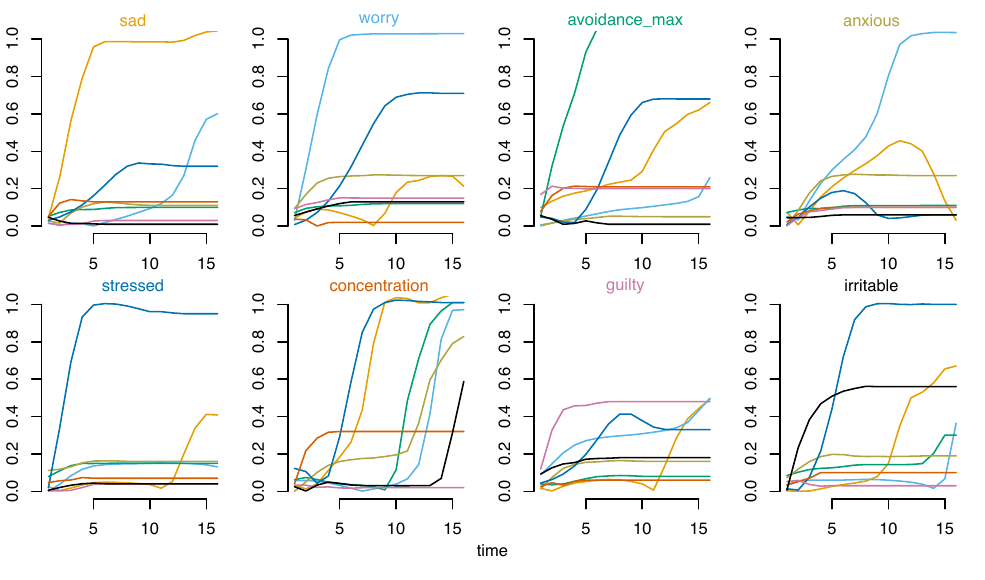}
\caption{ patient 4 }
\label{fig:intervention-effect-time-p4}
\end{figure}
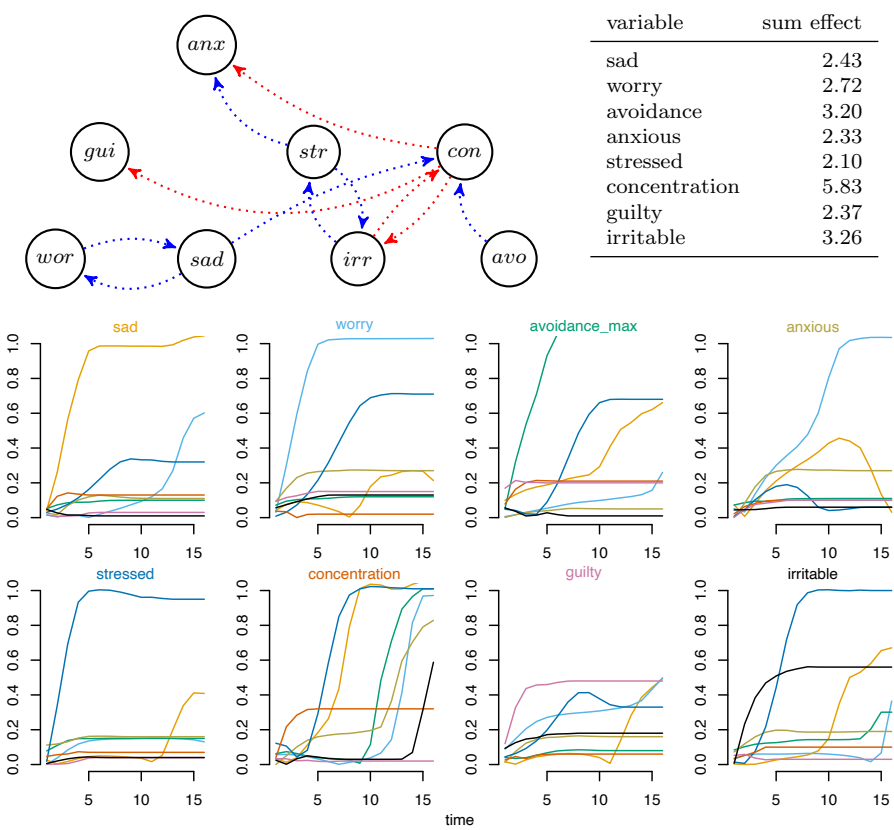
\begin{figure}[t]
\begin{tabular}{c @{\hspace{2em}} p{0.4\textwidth}}
\raisebox{-10em}{
\begin{tikzpicture}[-, >=stealth',shorten >=1pt,auto,node distance=2cm,
  thick,main node/.style={circle,fill=blue!10,draw,font=\footnotesize\sffamily}]
  \tikzstyle{sample} = [circle,fill=blue!20,draw,font=\footnotesize\sffamily]
  \tikzstyle{noSample} = [circle,fill=red!0,draw,font=\footnotesize\sffamily]
  \tikzstyle{sampleEdge} = [font=\sffamily\small,blue]
  \tikzstyle{noSampleEdge} = [node/.style={font=\sffamily\small},red]	
	
  \node[noSample] (anx) [] {$anx$};
  \node[noSample] (gui) [below left of=anx] {$gui$};
  \node[noSample] (str) [below right of=anx] {$str$};
  \node[noSample] (con) [right of=str] {$con$};
  \node[noSample] (sad) [below right of=gui] {$sad$};
  \node[noSample] (irr) [right of=sad] {$irr$};
  \node[noSample] (avo) [right of=irr] {$avo$};  
  \node[noSample] (wor) [left of=sad] {$wor$};

\path[sampleEdge,->,dotted,red]
 (con) edge [bend left=15] node [xshift=0.2em] {} (anx)
 (str) edge [bend left=30] node [xshift=0.2em] {} (con)
 (con) edge [bend left=30] node [xshift=0.2em] {} (avo)
 (con) edge [bend left=30] node [xshift=0.2em] {} (str)
 ;
 %above, yshift=0em, xshift=0.4em

 \path[sampleEdge,->,dotted]
 (sad) edge [bend right] node [xshift=0.2em] {} (irr)
 (wor) edge [bend left=20] node [xshift=0.2em] {} (irr)
 (avo) edge [bend left=20] node [xshift=0.2em] {} (wor)
 (avo) edge [bend left=30] node [xshift=0.2em] {} (str)
 (anx) edge [bend left=30] node [xshift=0.2em] {} (wor)
 (anx) edge [bend right=30] node [xshift=0.2em] {} (str)
 (anx) edge [bend left=30] node [xshift=0.2em] {} (irr)
 (str) edge [bend left=30] node [xshift=0.2em] {} (anx)
 (gui) edge [bend left=30] node [xshift=0.2em] {} (wor)
 (gui) edge [bend left=30] node [xshift=0.2em] {} (irr)
 (irr) edge [bend left=30] node [xshift=0.2em] {} (sad)
 (irr) edge [bend left=30] node [xshift=0.2em] {} (anx)
 (irr) edge [bend left=30] node [xshift=0.2em] {} (gui)
 ;
 %above, yshift=0em, xshift=0.4em
    
\end{tikzpicture}
}
&
\vspace{-3em}
\begin{tabular}{l @{\hspace{1em}} r}
variable    &sum effect\\ \midrule
sad         &1.12\\ 
worry       &1.41\\
avoidance   &0.73\\
anxious     &0.86\\
stressed    &1.24\\
concentration &1.09\\
guilty      &2.26\\
irritable   &2.97\\
\midrule
\end{tabular}
\end{tabular}
\\
\centering
    \pgfimage[width=\textwidth]{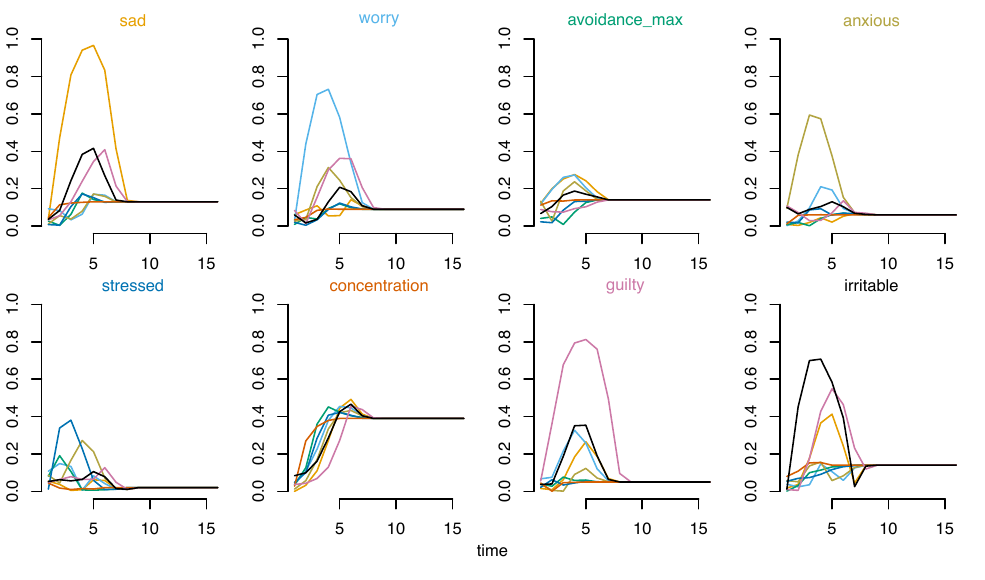}
\caption{ patient 5 }
\label{fig:intervention-effect-time-p5}
\end{figure}
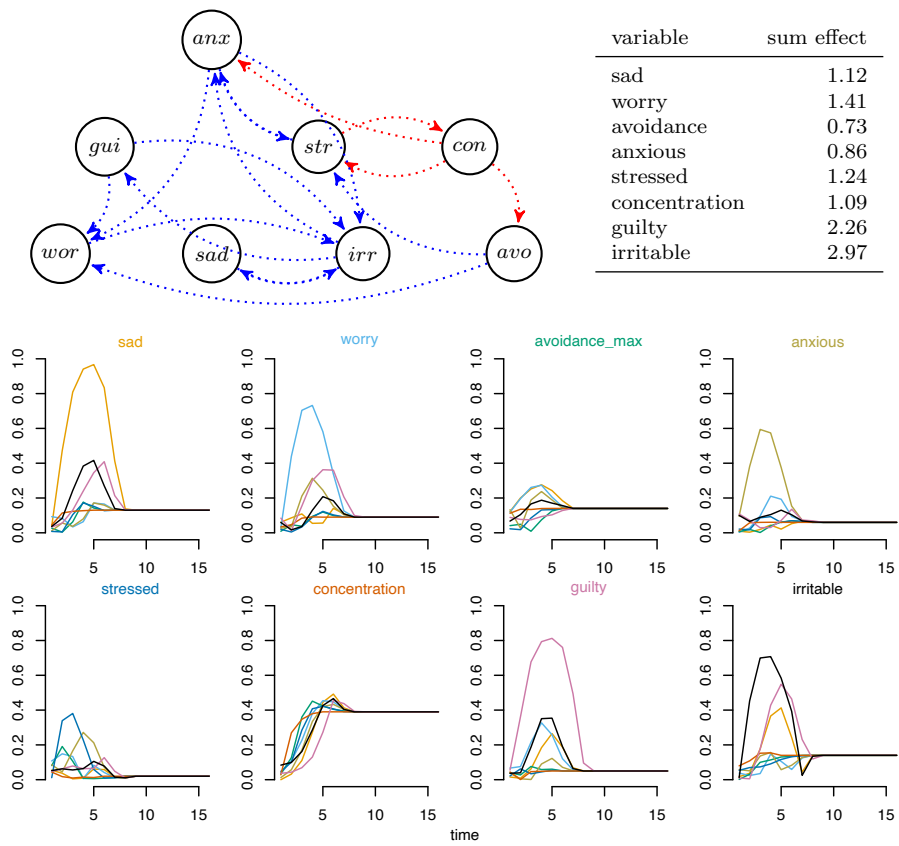
%

%---------------------------------------------------------------------------
\end{document}